\documentclass[runningheads,a4paper]{llncs}
\usepackage{llncsdoc}
\usepackage[utf8]{inputenc}
\usepackage[english]{babel}
\usepackage{gensymb}
\usepackage{stmaryrd}
\usepackage{colortbl}
\usepackage{enumitem}
\usepackage{tikz}
\usepackage{verbatim}
\usepackage{textcomp}
\usepackage{graphicx}
\usepackage{epsfig}
\usepackage{epstopdf}
\usepackage{color}
\usepackage{amsmath,amsfonts, amssymb}
\usepackage{placeins}
\usepackage{url}
\usepackage[numbers, comma, compress, sectionbib]{natbib}
\usepackage{float}
\usepackage[ruled,vlined,linesnumbered]{algorithm2e}
\usepackage{algorithmic}
\usepackage{pifont}
\usepackage{mathtools}
\usepackage{hyperref}
\usepackage{booktabs}
\SetKwRepeat{Do}{do}{while}%
\usepackage{etoolbox}
\usepackage{diagbox}
\usepackage{multirow}
\usepackage{makecell}

\usepackage{threeparttable}

\newcommand{\R}{\ensuremath{\mathbb{R}}}

\newcommand{\Z}{\ensuremath{\mathbb{Z}}}

\DeclarePairedDelimiter\inner{\langle}{\rangle}

\DeclarePairedDelimiter\set{\{}{\}}

\DeclarePairedDelimiter\floor{\lfloor}{\rfloor}
\DeclarePairedDelimiter\ceil{\lceil}{\rceil}

\newcommand{\matA}{\ensuremath{\mathbf{A}}}
\newcommand{\matB}{\ensuremath{\mathbf{B}}}

\newcommand{\matD}{\ensuremath{\mathbf{D}}}
\newcommand{\matE}{\ensuremath{\mathbf{E}}}
\newcommand{\matF}{\ensuremath{\mathbf{F}}}
\newcommand{\matG}{\ensuremath{\mathbf{G}}}

\newcommand{\matI}{\ensuremath{\mathbf{I}}}

\newcommand{\matP}{\ensuremath{\mathbf{P}}}
\newcommand{\matQ}{\ensuremath{\mathbf{Q}}}
\newcommand{\matR}{\ensuremath{\mathbf{R}}}
\newcommand{\matS}{\ensuremath{\mathbf{S}}}
\newcommand{\matT}{\ensuremath{\mathbf{T}}}

\newcommand{\matV}{\ensuremath{\mathbf{V}}}

\newcommand{\veca}{\ensuremath{\mathbf{a}}}
\newcommand{\vecb}{\ensuremath{\mathbf{b}}}
\newcommand{\vecc}{\ensuremath{\mathbf{c}}}
\newcommand{\vecd}{\ensuremath{\mathbf{d}}}
\newcommand{\vece}{\ensuremath{\mathbf{e}}}
\newcommand{\vecf}{\ensuremath{\mathbf{f}}}
\newcommand{\vecg}{\ensuremath{\mathbf{g}}}

\newcommand{\vecp}{\ensuremath{\mathbf{p}}}

\newcommand{\vecs}{\ensuremath{\mathbf{s}}}
\newcommand{\vect}{\ensuremath{\mathbf{t}}}
\newcommand{\vecu}{\ensuremath{\mathbf{u}}}
\newcommand{\vecv}{\ensuremath{\mathbf{v}}}

\newcommand{\vecx}{\ensuremath{\mathbf{x}}}
\newcommand{\vecy}{\ensuremath{\mathbf{y}}}
\newcommand{\vecz}{\ensuremath{\mathbf{z}}}
\newcommand{\veczero}{\ensuremath{\mathbf{0}}}

\usepackage{amsbsy}
\usepackage[mathscr]{eucal}

\newcommand{\cR}{\mathcal R}

\newcommand{\lamperp}{\Lambda^{\perp}}

\newcommand{\lat}{\Lambda}

\newcommand{\smooth}{\eta}
\newcommand{\smootheps}{\smooth_{\epsilon}}

\newcommand{\gs}[1]{\ensuremath{\widetilde{#1}}}

\DeclareMathOperator{\vol}{vol}

\DeclareMathOperator{\spn}{span}

\mathtoolsset{showonlyrefs}

\makeatletter
\renewcommand{\pod}[1]{\mathchoice
  {\allowbreak \if@display \mkern 18mu\else \mkern 8mu\fi (#1)}
  {\allowbreak \if@display \mkern 18mu\else \mkern 8mu\fi (#1)}
  {\mkern4mu(#1)}
  {\mkern4mu(#1)}
}
\makeatletter
\let\@@pmod\pmod
\DeclareRobustCommand{\pmod}{\@ifstar\@pmods\@@pmod}
\def\@pmods#1{\mkern4mu({\operator@font mod}\mkern 6mu#1)}
\makeatother

\hypersetup{
    colorlinks=true, 
    breaklinks=true,  
    urlcolor= blue,   
    linkcolor= blue,  
    citecolor= blue  
}

\addto\captionsenglish{}

\usepackage{algorithmic}

\newenvironment{proof*}[1]
{\renewcommand{\proofname}{#1}%
\begin{proof}}
{\end{proof}}

\newcommand{\orth}{\lamperp}
\newcommand{\Exp}{\mathbb{E}}

\newcommand{\algGadget}{\mathsf{ApproxGadget}}

\newcommand{\algPreSamp}{\mathsf{ApproxPreSamp}}
\newcommand{\algDecode}{\mathsf{LatticeDecoder}}

\newcommand{\Robin}{\textsc{Robin}}
\newcommand{\Eagle}{\textsc{Eagle}}
\newcommand{\algNtruKG}{\mathsf{Robin.KeyGen}}
\newcommand{\algNtruSign}{\mathsf{Robin.Sign}}
\newcommand{\algNtruVerify}{\mathsf{Robin.Verify}}
\newcommand{\algLWEKG}{\mathsf{Eagle.KeyGen}}
\newcommand{\algLWESign}{\mathsf{Eagle.Sign}}
\newcommand{\algLWEVerify}{\mathsf{Eagle.Verify}}
\newcommand{\hash}{\ensuremath{\mathsf{H}}}
\newcommand{\salt}{\ensuremath{\mathsf{salt}}}
\newcommand{\msg}{\ensuremath{\mathsf{msg}}}
\newcommand{\seed}{\ensuremath{\mathsf{seed}}}
\newcommand{\expand}{\mathsf{Expand}}

\usepackage{marvosym}

\title{Compact Lattice Gadget and Its Applications to Hash-and-Sign Signatures}
\author{Yang Yu\inst{1,2,3} \and Huiwen Jia\inst{4,5} \and Xiaoyun Wang\inst{6,7,8}}
\institute{
BNRist, Tsinghua University, Beijing, China\\
\email{yu-yang@mail.tsinghua.edu.cn}
\and
Zhongguancun Laboratory, Beijing, China\\
\and
National Financial Cryptography Research Center, Beijing, China\\
\and
School of Mathematics and Information Science, Key Laboratory of Information Security, Guangzhou University, Guangzhou, China\\
\email{hwjia@gzhu.edu.cn}
\and
Guangzhou Center for Applied Mathematics, Guangzhou University, Guangzhou, China\\
\and
Institute for Advanced Study, Tsinghua University, Beijing, China\\
\email{xiaoyunwang@mail.tsinghua.edu.cn}
\and 
Key Laboratory of Cryptologic Technology and Information Security (Ministry of Education), School of Cyber Science and Technology, Shandong University, Qingdao, China
\and
Shandong Institute of Blockchain, Jinan, China\\
}

\begin{document}
\maketitle

\begin{abstract}

Lattice gadgets and the associated algorithms are the essential building blocks of lattice-based cryptography.
In the past decade, they have been applied to build versatile and powerful cryptosystems.
However, the practical optimizations and designs of gadget-based schemes generally lag their theoretical constructions.
For example, the gadget-based signatures have elegant design and capability of extending to more advanced primitives, but they are far less efficient than other lattice-based signatures.

This work aims to improve the practicality of gadget-based cryptosystems, with a focus on hash-and-sign signatures.
To this end, we develop a compact gadget framework in which the used gadget is a \emph{square} matrix instead of the short and fat one used in previous constructions.
To work with this compact gadget, we devise a specialized gadget sampler, called \emph{semi-random sampler}, to compute the approximate preimage.
It first \emph{deterministically} computes the error and then randomly samples the preimage.
We show that for uniformly random targets, the preimage and error distributions are simulatable without knowing the trapdoor.
This ensures the security of the signature applications.
Compared to the Gaussian-distributed errors in previous algorithms,
the deterministic errors have a smaller size, which lead to a substantial gain in security and enables a practically working instantiation.

As the applications, we present two practically efficient gadget-based signature schemes based on NTRU and Ring-LWE respectively.
The NTRU-based scheme offers comparable efficiency to Falcon and Mitaka and a simple implementation without the need of generating the NTRU trapdoor.
The LWE-based scheme also achieves a desirable overall performance.
It not only greatly outperforms the state-of-the-art LWE-based hash-and-sign signatures, but also has an even smaller size than the LWE-based Fiat-Shamir signature scheme Dilithium.
These results fill the long-term gap in practical gadget-based signatures.

\end{abstract}

\section{Introduction}\label{sec:intro}
Lattice-based cryptography is a promising post-quantum cryptography family having attractive features in both theory and practice.
It has been shown to provide powerful versatility leading to various advanced cryptosystems including fully homomorphic encryption~\cite{gentry2009fully}, attribute-based encryption~\cite{gorbunov2013attribute},
group signatures~\cite{gordon2010group} and
much more~\cite{gorbunov2015predicate,brakerski2016obfuscating,agrawal2017stronger,peikert2019noninteractive}.
For the basic encryption and signatures, lattice-based schemes are the most practically efficient among post-quantum cryptosystems and
three of four post-quantum algorithms selected by NIST for standardization are lattice-based: Kyber~\cite{Kyber} for public key encryption/KEMs; Dilithium~\cite{dilithium} and Falcon~\cite{falcon} for digital signatures.

At the core of many lattice-based schemes is the so-called Ajtai's function $f_\matA(\vecx) = \matA\vecx \bmod Q$ where $\matA \in \Z_Q^{n \times m}$ is a short and fat random matrix.
Ajtai showed in his seminal work~\cite{ajtai1996generating} that the inversion of $f_\matA$, i.e.
finding a short preimage $\vecx$, is as hard as some worst-case lattice problems.
With a \emph{lattice trapdoor} for $\matA$, one can efficiently compute a short preimage.
In some applications, e.g. signatures, the preimage distribution is required to be \emph{simulatable} without knowing the trapdoor.
This is essential for security: some early proposals~\cite{GGH,NTRUSign} were indeed broken by statistical attacks~\cite{nguyen2006learning,ducas2012learning,yu2018learning}, since the preimages leak information of the trapdoor.
To get rid of such leaks, Gentry, Peikert and Vaikuntanathan proposed a provably secure trapdoor framework, known as the GPV framework~\cite{gentry2008trapdoors}, in which the preimage is sampled from a distribution statistically close to some publicly known discrete Gaussian.
In the past decade, the GPV framework has been continuously enriched by new Gaussian sampling algorithms and trapdoor constructions.
This leads to a series of efficient instantiations that can be basically classified into two families: \emph{NTRU trapdoor based} and \emph{gadget based}.

\subsubsection{NTRU trapdoor based GPV instantiations.}
The NTRU trapdoor, that is a high-quality basis of the NTRU lattice, was originally used in~\cite{NTRUSign}.
In~\cite{ducas2014efficient}, Ducas, Lyubashevsky and Prest first discovered that
the lengths of the NTRU trapdoors can be within a small constant factor of optimal by choosing proper parameters, which gives a compact instantiation of the GPV framework over NTRU lattices. As an application, they presented the first lattice-based identity-based encryption (IBE) scheme with practical parameters.
This instantiation was further developed as the Falcon signature scheme by integrating the fast Fourier sampler~\cite{ducas2016fast}.
Falcon is now selected by NIST for the post-quantum standaradization, due to its good performance in terms of bandwidth and efficiency.
However, the signing and key generation algorithms of Falcon are rather complex.
Recently, Espitau et al. proposed a simplified variant of Falcon, called Mitaka~\cite{espitau2022simpler}. Mitaka uses the hybrid sampler~\cite{prest2015phd} for easier implementation at the cost of a substantial security loss.
To mitigate the security loss, Mitaka adopts some techniques to improve the trapdoor quality, which further complicates the key generation.
Overall, Falcon and Mitaka are currently the most efficient lattice-based signatures, but their complex algorithms may be difficult to implement in constrained environments.
Furthermore, the security of NTRU is shown to be significantly reduced in the overstretched parameter regime~\cite{kirchner2017revisiting,ducas2021ntru}, thus
the NTRU trapdoor based instantiations are mainly used in signature and IBE applications.

\subsubsection{Gadget based GPV instantiations.}
The gadget based instantiation was first proposed by Micciancio and Peikert~\cite{micciancio2012trapdoors}.
In the Micciancio-Peikert framework, the public matrix $\matA = [\bar{\matA} \mid \matG - \bar{\matA}\matR]$ where the trapdoor $\matR$ is a matrix with small entries and the \emph{gadget} $\matG = \matI_n \otimes \vecg^t$ with $\vecg = (1, b,\cdots,b^{k-1}),k = \ceil{\log_b(Q)}$.
The inversion of $f_\matA$ is converted into the inversion of $f_\matG$ with $\matR$.
The latter boils down to the Gaussian sampling over the lattice $\orth_Q(\vecg) = \set{\vecu \mid \inner{\vecu, \vecg} = 0 \bmod Q}$ that is easy and fast~\cite{micciancio2012trapdoors,genise2018faster,zhang2022towards}.
Compared to the NTRU trapdoor based GPV instantiations, the gadget based framwork offers significant advantages in terms of implementation and turns out to be extremely versatile for the constructions of advanced primitives.
However, the gadget based schemes suffer from rather large preimage and public key sizes.
To improve the practicality of gadget based schemes, Chen, Genise and Mukherjee introduced the notion of approximate trapdoor~\cite{chen2019approximate} and proposed to use a truncated gadget $\vecf = (b^l,\cdots,b^{k-1})$ for the trapdoor construction.
While the improvement is substantial, the size of their gadget-based signature scheme is still far larger than that of Falcon and Dilithium.

As seen above, both NTRU trapdoor based schemes and gadget based ones occupy fairly different positions in lattice-based cryptography, but they also have own limitations.
Particularly, the practical designs of gadget-based cryptosystems still lag far behind their theoretical constructions.
It is therefore important to improve the practical efficiency of gadget-based cryptosystems including the hash-and-sign signatures.

\subsubsection{Our Contributions.}
We develop some new technique to reduce the size of the gadget-based schemes.
Using our compact gadget, we propose two hash-and-sign signature schemes based on NTRU and Ring-LWE respectively.
They both offer a desirable performance and an easy implementation.
This fills the gap in practical gadget-based signatures.

\paragraph{Compact gadget with semi-random sampler.}
In our construction, the used gadget is $\matP \in \Z^{n\times n}$ along with $\matQ \in \Z^{n\times n}$ such that
$\matP\matQ = Q\cdot \matI_n$, and the trapdoor $\matT$ for the public matrix $\matA \in \Z_Q^{n\times m}$ satisfies $\matA\matT=\matP\mod Q$.
The main technique to enable this compact gadget is a new gadget sampler for approximate trapdoors, called \emph{semi-random} sampler.
Given the target $\vecu$, this sampler computes a short approximate preimage $\vecx$ such that
$\vecu = \matP\vecx + \vece \bmod Q$ with a short error $\vece$.
In our sampler, only the preimage is randomly generated and the error is fixed by the target, which is why we name ``semi-random''.
More concretely, the semi-random sampler consists of two steps respectively performed over the lattices defined by $\matP$ and $\matQ$:
\begin{enumerate}
	\item \emph{Deterministic error decoding:} The sampler first computes an error $\vece$ such that $\vecu-\vece = \matP\vecc \in \lat(\matP)$ with deterministic lattice decoding.
	\item \emph{Random preimage sampling:} Then the sampler generates a short preimage $\vecx \in \lat(\matQ)+\vecc$ with Gaussian sampling.
\end{enumerate}
It is easy to verify that $\matP\vecx = \vecu - \vece \bmod Q$.
Despite the deterministic errors, we show that the distribution of $(\vecx,\vece)$ is simulatable for \emph{uniformly random targets}. This is sufficient for the applications of digital signatures.

Our general construction can be instantiated with various lattices $\lat(\matP)$ and $\lat(\matQ)$ with specialized decoding and sampling algorithms.
This opens up interesting avenues in the designs of lattice trapdoors.
This paper showcases the merit of our gadget construction with a natural and simple instantiation: $(\matP = p\matI_n, \matQ = q\matI_n)$.
We now contrast this simple instantiation with the truncated gadget $ \matI_n \otimes  (b^l,\cdots,b^{k-1})^t$ in~\cite{chen2019approximate}.
Indeed, the gadget in our instantiation has the same structure with the special case of~\cite{chen2019approximate} in which $l = k-1$, $b^l = p$ and $Q=b^{l+1}$,
but the associated sampling algorithms in two cases are quite distinct, which yields the differences in size and efficiency (see Table~\ref{tab:sampler}).
For uniformly random targets, the error in our gadget sampler is uniformly distributed over $\Z_p^n$ and the preimage is distributed as Gaussian of width $q\cdot \omega(\sqrt{\log n})$.
Then the error size is $\approx \frac{p\sqrt{n}}{\sqrt{12}}$ and the preimage size is $\approx \frac{Q}p\cdot\omega(\sqrt{n\log n})$.
When it comes to the case of~\cite{chen2019approximate},
the error and the preimage are distributed as Gaussian of width $\sigma\sqrt{\frac{b^{2l}-1}{b^2-1}}$ and $\sigma$ respectively where $\sigma \geq \sqrt{b^2+1}\cdot\omega(\sqrt{\log n})$.
Then the error size is $\approx p\cdot\omega(\sqrt{n\log n})$ and the preimage size is $\approx \frac{Q}p\cdot\omega(\sqrt{n\log n})$ when $b^l=p$ and $Q = b^{l+1}$.
As a consequence, our technique reduces the error size by a factor of $\sqrt{12}\cdot\omega(\sqrt{\log n})$ while keeping the preimage size.
This gives a noticable gain in concrete security and enables a practically working instantiation with the compact gadget.
In addition, our semi-random sampler only needs $n$ times integer Gaussian sampling along with $n$ times modulo operations, whereas
the sampler in~\cite{chen2019approximate} needs $nk$ times integer Gaussian sampling along with $O(nk)$ additions and multiplications.
Our sampler is therefore simpler and more efficient.
To sum up, our technique substantially improves the practical performance of the gadget-based schemes.
\begin{table}[H]\renewcommand\arraystretch{1.5}
	\centering
	\caption{Comparisons with our gadget sampler with the special case of~\cite{chen2019approximate} in which $l = k-1$, $b^l = p$ and $Q=b^{l+1}$.}\label{tab:sampler}
	\setlength{\tabcolsep}{0.5mm}{}
	\begin{tabular}{p{2cm}<{\centering} p{2.5cm}<{\centering} p{2.5cm}<{\centering} p{3cm}<{\centering}}
		\hline
		& preimage size & error size&  $\#$integer sampling \\
		\hline
		\cite{chen2019approximate}  & $\frac{Q}p\cdot\omega(\sqrt{n\log n})$ & $p\cdot\omega(\sqrt{n\log n})$ & $nk$\\
		This work & $\frac{Q}p\cdot\omega(\sqrt{n\log n})$ & $\frac{p\sqrt{n}}{\sqrt{12}}$ &  $n$ \\
		\hline
	\end{tabular}
\end{table}

\paragraph{Simpler NTRU-based hash-and-sign signatures.}
We use the new gadget algorithms to build a new NTRU-based hash-and-sign signature scheme \Robin.
It achieves high efficiency comparable to Falcon~\cite{falcon} and Mitaka~\cite{espitau2022simpler} that are two representative NTRU trapdoor based signatures (see Table~\ref{tab:NTRU}).
The main advantage of \Robin{} is its convenient implementation.
Firstly, \Robin{} uses one NTRU vector instead of a full NTRU trapdoor basis as the signing key, which avoids the highly complex key generation.
Secondly, like most of gadget-based signatures, the signing procedure of \Robin{} has an online/offline structure and the online sampling only consists of $D_{8\Z+c,r}$ for $c =0,1,\cdots,7$, which allows an easier and more efficient implementation and side-channel protection.
Additionally, the whole \Robin{} algorithm including key generation, signing and verification can be conveniently implemented without using floating-point arithmetic.
Therefore, \Robin{} can be seen as an attractive post-quantum signature scheme especially for constrained devices.
\begin{table*}[h]\renewcommand\arraystretch{1.2}	
	\centering
	\caption{Comparison between \Robin{} with Falcon~\cite{falcon} and Mitaka~\cite{espitau2022simpler} at NIST-I and NIST-V security levels. }
	\label{tab:NTRU}
	\centerline{{
			\setlength\tabcolsep{2pt}
			\begin{tabular}{rcccccc}
				\toprule
				& \multicolumn{3}{c}{NIST-I level} &
				\multicolumn{3}{c}{NIST-V level}
				\\
				\cmidrule(lr){2-4}\cmidrule(lr){5-7}
				&{Falcon} & {Mitaka} &
				{\Robin} & {Falcon} &
				{Mitaka} &  {\Robin} \\
				\midrule
				Sig.{} size (bytes)  & 643  & 807 & 992 & 1249 & 1376 &  1862 \\
				Pub.{} key size (bytes)  & 896  & 972 & 1227  & 1792 & 1792 & 2399   \\
				\bottomrule
			\end{tabular}
	}}
\end{table*}

\paragraph{Shorter LWE-based hash-and-sign signatures.}
We also propose a Ring-LWE-based instantiation of signatures based on our gadget, called \Eagle.
While \Eagle{} is less efficient than its NTRU-based counterpart \Robin, it still has a desirable performance and a simple implementation.
Compared to other LWE-based hash-and-sign signatures, \Eagle{} offers a significantly smaller bandwidth.
Specifically, the signature (resp. public key) size of \Eagle{} is $\leq 55\%$ (resp. $\leq 35\%$) of that of the scheme from~\cite{chen2019approximate} with refined parameters and security estimates for both 80-bits and 192-bits of security levels.
In fact, \Eagle{} is even more compact than Dilithium that is a representative LWE-based Fiat-Shamir signature scheme: 
for 192-bits of security level, the signature size of \Eagle{} is smaller by $\approx 8\%$ compared with Dilithium. 
To the best of our knowledge, \Eagle{} is the first LWE-based hash-and-sign signature scheme of key and signature sizes on par or better than practical LWE-based Fiat-Shamir signatures.

	
	\begin{table*}[h]\renewcommand\arraystretch{1.2}	
		\centering
		\caption{Comparison between \Eagle{} with~\cite{chen2019approximate} at $80$-bits and NIST-III ($192$-bits) security levels. }
		\label{tab:LWE}
		\centerline{{
				\setlength\tabcolsep{4pt}
				\begin{tabular}{rcccc}
					\toprule
					& \multicolumn{2}{c}{80-bits security} &
					\multicolumn{2}{c}{NIST-III level}
					\\
					\cmidrule(lr){2-3}\cmidrule(lr){4-5}
					&{\cite{chen2019approximate}}  &
					{\Eagle} & {\cite{chen2019approximate}} &
					{\Eagle} \\
					\midrule
					Sig.{} size (bytes)  & 2753    & 1406   & 7172    & 3052    \\
					Pub.{} key size (bytes)& 2720  & 928   & 7712    & 1952    \\
					\bottomrule
				\end{tabular}
		}}
	\end{table*}

	\subsubsection{Roadmap.}
	We start in Section~\ref{sec:prelim} with preliminary materials, followed by recalling the existing gadget trapdoors in Section~\ref{sec:recall}.
	Section~\ref{sec:gadget} introduces our new gadget and the corresponding approximate trapdoor framework.
	We present concrete NTRU-based and Ring-LWE-based hash-and-sign signatures instantiated with our compact gadget framework in Section~\ref{sec:ntru} and Section~\ref{sec:lwe} respectively.
	Finally, we conclude in Section~\ref{sec:conclusion}.

\section{Preliminaries}\label{sec:prelim}

\paragraph{Notations}
Let $\mathbb{R}$ and $\mathbb{Z}$ denote the set of real numbers and integers respectively.
For a positive integer $q$, let $\Z_q = \set{-\floor{q/2}, -\floor{q/2}+1,\cdots,q-\floor{q/2} -1}$.
For a real-valued function $f$ and a countable set $S$, we write $f(S) = \sum_{x\in S} f(x)$
assuming this sum is absolutely convergent.
We write $a \gets D$ to represent the sample $a$ drawn from the distribution $D$. For a finite set $S$, let $U(S)$ be the uniform distribution over $S$ and $a \overset{\$}{\leftarrow} S$ denote the sample $a\gets U(S)$.

\subsection{Linear algebra and lattices}
A vector is denoted by
a bold lower case letter, e.g. $\vecx=(x_1,\ldots,x_n)$, and in column form. The concatenation of $\vecx_1,\vecx_2$ is denoted by $(\mathbf x_1,\mathbf x_2)$.
Let $\inner{\vecx,\vecy}$ be the inner product of $\vecx,\vecy \in \R^n$ and
$\|\vecx\| = \sqrt{\inner{\vecx,\vecx}}$ be the $\ell_2$ norm of $\vecx$.
A matrix is denoted by a bold upper case letter, e.g. $\matA=[\veca_1\mid\cdots\mid\veca_n]$, where $\veca_i$ denotes the $i^{th}$ column of $\mathbf A$.
Let $\gs{\matA} = [\gs{\veca_1}\mid\cdots\mid\gs{\veca_n}]$ denote the Gram-Schmidt orthogonalization of $\matA$.
Let $\mathbf A\oplus\mathbf B$ denote the block diagonal concatenation of $\mathbf A$ and $\mathbf B$. The largest singular value of $\mathbf A$ is denoted by $s_1(\mathbf A)=\max_{\mathbf{x\neq0}}\frac{\|\mathbf{Ax}\|}{\|\mathbf x\|}$. Let $\matA^t$ be the transpose of $\matA$.

We write $\Sigma\succ0$, when a symmetric matrix $\Sigma\in\mathbb{R}^{m\times m}$ is positive definite, i.e. $\mathbf{x}^t\Sigma\mathbf{x}>0$ for all nonzero $\mathbf{x}\in\mathbb{R}^m$. We write $\Sigma_1\succ\Sigma_2$ if $\Sigma_1-\Sigma_2\succ0$. For any scalar $s$, we write $\Sigma\succ s$ if  $\Sigma-s\cdot\mathbf I\succ0$. If $\Sigma=\mathbf{BB}^t$, we call $\mathbf{B}$ a square root of $\Sigma$. We use $\sqrt{\Sigma}$ to denote any square root of $\Sigma$ when the context permits it.


Given $\mathbf B=[\mathbf b_1\mid\cdots\mid\mathbf b_n]\in\mathbb R^{m\times n}$ with each $\vecb_i$ linearly independent, the lattice generated by $\matB$ is
$\lat(\matB) = \set{\matB\vecz \mid \vecz \in \Z^n}$. The dimension of $\lat(\matB)$ is $n$ and $\matB$ is called a basis.
Let $\lat^* = \set{\vecy \in \spn(\lat) \mid \inner{\vecx,\vecy} \in \Z, \forall \vecx \in \lat}$ be the dual lattice of a lattice $\lat$.

In lattice-based cryptography, the $q$-ary lattice is of special interest and defined for some $\matA \in \Z_q^{n\times m}$ as
\begin{equation*}
\Lambda^{\perp}_q(\mathbf{A})  = \{\mathbf{x}\in\mathbb{Z}^m: \mathbf{A x}=\mathbf{0}\!\!\mod q\}.
\end{equation*}
The dimension of $\Lambda^{\perp}_q(\mathbf{A})$ is $m$ and $(q\cdot\Z)^m \subseteq  \Lambda^{\perp}_q \subseteq \Z^m$.
Each $\vecu \in \Z^n_q$ defines a lattice coset
\begin{equation*}
	\Lambda^{\perp}_{q,\vecu}(\mathbf{A})  = \{\mathbf{x}\in\mathbb{Z}^m: \mathbf{A x}=\mathbf{u}\!\!\mod q\}.
\end{equation*}

Given a matrix $\matA\in\Z_q^{n\times m}$, let $f_\matA(\vecx) = \matA\vecx \bmod q$ be the associated Ajtai's function~\cite{ajtai1996generating} where $\vecx$ is usually short. We simply denote by $f^{-1}_\matA$ the inversion procedure, namely finding a short \emph{preimage} $\vecx$.

\subsection{Gaussians}
The Gaussian function $\rho:\mathbb{R}^m\rightarrow(0,1]$ is defined as $\rho(\mathbf{x})=\exp(-\pi\cdot\langle\mathbf{x},\mathbf{x}\rangle)$. Applying a linear transformation given by an invertible matrix $\mathbf{B}$ yields
$$
\rho_{\mathbf{B}}(\mathbf{x})=\rho(\mathbf{B}^{-1}\mathbf x)=\exp(-\pi\cdot\mathbf{x}^t\Sigma^{-1}\mathbf{x}),
$$
where $\Sigma=\mathbf{BB}^t$.
Since $\rho_{\mathbf{B}}$ is exactly determined by $\Sigma$, we also write it as $\rho_{\sqrt{\Sigma}}$.
For a lattice $\Lambda$ and $\vecc \in \spn(\Lambda)$,
the discrete Gaussian distribution $D_{\Lambda+\vecc,\sqrt{\Sigma}}$ is defined as: for any $\vecx \in \Lambda + \vecc$,
\[ D_{\Lambda+\vecc,\sqrt{\Sigma}}(\vecx)=\frac{\rho_{\sqrt{\Sigma}}(\vecx)}{\rho_{\sqrt{\Sigma}}(\Lambda+\vecc)}.
\]


Let
$\smootheps(\Lambda) = \min\set{s > 0 \mid \rho(s\cdot\Lambda^*)\leq1+\epsilon}$ be the smoothing parameter with respect to a lattice $\lat$ and $\epsilon\in(0,1)$.
We write $\sqrt{\Sigma}\geq\eta_{\epsilon}(\Lambda)$, if $\rho_{\sqrt{\Sigma^{-1}}}(\Lambda^*)\leq1+\epsilon$.


\begin{lemma}[\cite{gentry2008trapdoors}]\label{lemma:etabound}
	Let $\Lambda$ be an $m$-dimensional lattice with a basis $\mathbf B$, then
	$
	\eta_{\epsilon}(\Lambda)\leq
	\max_i \|\gs{\vecb_i}\|\cdot\sqrt{\log{(2m(1+1/\epsilon))}/\pi},
	$
    where $\gs{\vecb_i}$ is the $i$-th vector of $\gs{\mathbf B}$.
\end{lemma}

\begin{lemma}[\cite{micciancio2007worst}]\label{lem:smooth}
	Let $\Lambda$ be a lattice, $\vecc\in \spn(\Lambda)$.
	Then for any $\epsilon\in(0,\frac{1}{2})$ and $s\geq\smootheps(\Lambda)$, $\rho_{s}(\Lambda + \vecc) \in \left[\frac{1-\epsilon}{1+\epsilon}, 1\right]\rho_{s}(\Lambda)$.
\end{lemma}

\begin{lemma}[\cite{gentry2008trapdoors}, Corollary 2.8]\label{lem:smooth_uniform}
	Let $\Lambda,\Lambda'$ be two lattices such that $\Lambda'\subseteq\Lambda$.
	Let $s\geq\eta_{\epsilon}(\Lambda')$. Then for any $\epsilon\in(0,\frac{1}{2})$ and $\mathbf c\in \spn(\Lambda)$, the distribution of $(D_{\Lambda+\vecc,s}\mod\Lambda')$ is within statistical distance at most $2\epsilon$ of $U(\Lambda\mod\Lambda')$.
\end{lemma}

\begin{theorem}[\cite{genise2020improved}]\label{thm:lineartransform}
    For any $\epsilon\in[0,1)$ defining $\bar{\epsilon}=2\epsilon/(1-\epsilon)$, a matrix $\mathbf S$ of full column rank, a lattice coset $A=\Lambda+\mathbf a\subset\mathrm{span}(\matS)$, and a matrix $\matT$ such that $\ker(\matT)$ is a $\Lambda$-subspace and $\eta_{\epsilon}(\Lambda\cap\ker(\matT))\leq\mathbf S$, we have
    $$
    \matT\cdot D_{A,\matS}\approx_{\bar{\epsilon}}D_{\matT A,\matT\matS}.
    $$
\end{theorem}

\subsection{The ring $\Z[x]/(x^n\pm 1)$}
We work with two polynomial rings in the paper.
The first one is the convolution ring $\cR_{n}^{-} = \Z[x]/(x^n - 1)$ where $n$ is a prime.
For any $a = \sum_{i=0}^{n-1} a_ix^i \in \cR_{n}^{-}$,
let $v(a) = (a_0,a_1\cdots,a_{n-1})$ be its coefficient vector and
the circulant matrix
\[
\mathcal{M}(a) = \begin{bmatrix}
	a_0 & a_{n-1} & \cdots & a_1\\
	a_1 & a_{0} & \cdots & a_2\\
	\vdots & \vdots & \vdots & \vdots \\
	a_{n-1} & a_{n-2} & \cdots & a_0 \\
\end{bmatrix} = [v(a),v(a\cdot x),\cdots, v(a\cdot x^{n-1})].
\]
be its matrix form.
The second ring in the paper is the power-of-$2$ cyclotomic ring, i.e.
$\cR_{n}^{+} = \Z[x]/(x^n + 1)$ with $n$ a power of $2$. For $a = \sum_{i=0}^{n-1} a_ix^i \in \cR_{n}^{+}$, its coefficient vector is also written as $v(a)$ and the matrix form becomes an anticirculant matrix
\[
\mathcal{M}(a) = \begin{bmatrix}
	a_0 & -a_{n-1} & \cdots & -a_1\\
	a_1 & a_{0} & \cdots & -a_2\\
	\vdots & \vdots & \vdots & \vdots \\
	a_{n-1} & a_{n-2} & \cdots & a_0 \\
\end{bmatrix} = [v(a),v(a\cdot x),\cdots, v(a\cdot x^{n-1})].
\]
In the rest of the paper, we identify $a$ with $v(a)$ when the context is clear.

Let $\bar a = a(x^{-1})$ for $a \in \cR$, then $\bar a = a_0 + \sum_{i=1}^{n-1} a_{n-i}x^i$ when $\cR = \cR_{n}^{-}$ and $\bar a = a_0 - \sum_{i=1}^{n-1} a_{n-i}x^i$ when $\cR = \cR_{n}^{+}$. More generally, let $\sigma_k(a) = a(x^k)$ for $k \in \Z_n^*$.
For both $\cR_{n}^{-}$ and $\cR_{n}^{+}$, the following properties hold:
\begin{itemize}
	\item $\mathcal{M}(a)+\mathcal{M}(b) = \mathcal{M}(a+b)$
	\item $\mathcal{M}(a)\cdot\mathcal{M}(b) = \mathcal{M}(ab)$.
	\item $\mathcal{M}(\bar a) = \mathcal{M}(a)^t$
\end{itemize}

\subsection{NTRU}
The NTRU module determined by $h \in \cR$ is given by
$$\lat^h_{NTRU}=\{(u,v)\in\mathcal R^2:uh-v=0\mod Q\}.$$
Our NTRU-based scheme mainly uses $\cR=\cR_n^-$, and the NTRU module is seen as a lattice of dimension $2n$.

In typical NTRU-based cryptosystems, the secret key is composed of two short polynomials $f,g \in \cR$, while the public key is $h = f^{-1}g\mod Q$. Then $(f,g)$ is a short vector of $\lat^h_{NTRU}$. In addition, an \emph{inhomogeneous} version of NTRU was introduced in~\cite{genise2019homomorphic}. In this version, the public key $h = f^{-1}(g+e)\bmod Q$ where $e$ is a public constant.
The corresponding problems are defined as follows.

\begin{definition}[NTRU and inhomogeneous NTRU]\label{def:ntru}
	Let $\cR=\Z[x]/(x^n - 1)$ with $n$ a prime.
	Let $Q > 0$ be an integer and $\chi$ be a distribution over $\cR$.
	Let $D_{\chi}$ (resp. $D_{\chi,e}$) be the distribution of the NTRU public key $h = f^{-1}g\mod Q$ (resp. $h = \frac{g+e}{f} \bmod Q$) with $f,g\gets \chi$.
	\begin{itemize}
		\item $\mathsf{NTRU}_{\cR,Q,\chi}$: Given $h \gets D_{\chi}$, find short $(f,g)$ such that $h = f^{-1}g\mod Q$.
		\item $\mathsf{iNTRU}_{\cR,Q,\chi,e}$: Given $h \gets D_{\chi,e}$, find short $(f,g)$ such that $h = \frac{g+e}{f}\mod Q$.
	\end{itemize}
\end{definition}

\subsection{LWE}
The LWE (learning with errors) problem is defined as follows.

\begin{definition}[LWE]
	Let $n,m,Q>0$ be integers and $\chi$ be a distribution over $\Z$. Given $\vecs \in \Z_Q^n$, let $A_{\vecs,\chi}$ be the distribution of $(\veca, b)$ where $\veca \overset{\$}{\leftarrow} \Z_Q^n$ and $b = \inner{\veca,\vecs}+e\bmod Q$ with $e \gets \chi$.
	\begin{itemize}
		\item Decision-$\mathsf{LWE}_{n,m,Q,\chi}$:
		Given $m$ independent samples from either $A_{\vecs,\chi}$ with $\vecs\gets \chi$ (fixed for all $m$ samples) or $U(\Z_Q^n\times\Z_Q)$, distinguish which is the case.
		\item Search-$\mathsf{LWE}_{n,m,Q,\chi}$: Given $m$ independent samples from $A_{\vecs,\chi}$ with $\vecs\gets \chi$, find $\vecs$.
	\end{itemize}
\end{definition}

To improve the efficiency and key sizes, some algebraic variants of LWE were proposed and used to build practical lattice-based cryptosystems. In this paper, we mainly use the ring variant proposed in~\cite{lyubashevsky2010ideal}.
\begin{definition}[Ring-LWE]
	Let $\cR = \Z[x]/(x^n+1)$ with $n$ a power of $2$.
	Let $m,Q>0$ be integers and $\chi$ be a distribution over $\cR$.
	Let $\cR_Q = \cR / (Q\cdot\cR)$.
	Given $s \in \cR_Q$, let $A_{s,\chi}$ be the distribution of $(a, b)$ where $a \overset{\$}{\leftarrow} \cR_Q$ and $b = as+e\bmod Q$ with $e \gets \chi$.
	\begin{itemize}
		\item Decision-$\mathsf{RLWE}_{\cR,m,Q,\chi}$:
		Given $m$ independent samples from either $A_{s,\chi}$ with $s\gets \chi$ (fixed for all $m$ samples) or $U(\cR_Q\times\cR_Q)$, distinguish which is the case.
		\item Search-$\mathsf{RLWE}_{\cR,m,Q,\chi}$: Given $m$ independent samples from $A_{s,\chi}$ with $s\gets \chi$, find $s$.
	\end{itemize}
\end{definition}

\subsection{SIS}
We recall the SIS (short integer solution) problem and its inhomogeneous variant.
\begin{definition}[SIS and inhomogeneous SIS]
	Let $n,m,Q>0$ be integers and $\beta > 0$.
	\begin{itemize}
		\item $\mathsf{SIS}_{n,m,Q,\beta}$: Given a uniformly random $\matA \in \Z_Q^{n\times m}$, find a non-zero integer vector $\vecx$ such that $\matA\vecx = \veczero \bmod Q$ and $\|\vecx\| \leq \beta$.
		\item $\mathsf{ISIS}_{n,m,Q,\beta}$: Given a uniformly random $\matA \in \Z_Q^{n\times m}$ and $\vecy \in \Z_Q^n$, find a non-zero integer vector $\vecx$ such that $\matA\vecx = \vecy \bmod Q$ and $\|\vecx\| \leq \beta$.
	\end{itemize}
\end{definition}

The public matrix $\matA$ in SIS and ISIS problems can be in the Hermite normal form (HNF), i.e. $\matA = [\matI_n \mid \matA']$.
This gives the HNF version of SIS problems, $\mathsf{HNF.SIS}$ and $\mathsf{HNF.ISIS}$. Such variants are as hard as the standard version.

The ring variants of SIS and ISIS are immediate. We only show the definition of Ring-ISIS.
\begin{definition}[Ring-ISIS, $\mathsf{RISIS}_{\cR,m,Q,\beta}$]
	Let $\cR = \Z[x]/(x^n+1)$ with $n$ a power of $2$.
	Let $m,Q>0$ be integers and $\beta > 0$.
	Let $\cR_Q = \cR / (Q\cdot\cR)$.
	Given a uniformly random $\matA \in \cR_Q^{m}$ and $y \in \cR_Q$, find a non-zero integer vector $\vecx \in \cR^m$ such that $\matA\vecx = y \bmod Q$ and $\|\vecx\| \leq \beta$.
\end{definition}

The approximate version of ISIS was introduced in~\cite{chen2019approximate}. It can be immediately adapted to the ring version $\mathsf{ApproxRISIS}_{\cR,m,Q,\alpha,\beta}$ and the HNF version $\mathsf{HNF.ApproxISIS}_{\cR,m,Q,\alpha,\beta}$.
\begin{definition}[Approximate ISIS, $\mathsf{ApproxISIS}_{n,m,Q,\alpha,\beta}$]
	Let $n,m,Q>0$ be integers and $\beta > 0$. Given a uniformly random $\matA \in \Z_Q^{n\times m}$ and a random $\vecy \in \Z_Q^n$, find an integer vector $\vecx$ such that $\matA\vecx = \vecy - \vece \bmod Q$ with $\|\vece\| \leq \alpha$ and $\|\vecx\| \leq \beta$.
\end{definition}

We will also use an NTRU version of SIS. It is the underlying assumption of NTRU-based signatures~\cite{falcon,espitau2022simpler,ducas2013lattice}. The NTRU-SIS problem can be immediately adapted to the inhomogeneous version $\mathsf{NTRUISIS}_{\cR,Q,\chi,\beta}$ and the approximate version $\mathsf{ApproxNTRUISIS}_{\cR,Q,\chi,\alpha,\beta}$.
\begin{definition}[NTRU-SIS, $\mathsf{NTRUSIS}_{\cR,Q,\chi,\beta}$]
	Let $\cR=\Z[x]/(x^n - 1)$ with $n$ a prime.
	Let $Q > 0$ be an integer, $\chi$ be a distribution over $\cR$ and $\beta > 0$.
	Given a random NTRU public key $h$ of either $\mathsf{NTRU}_{\cR,Q,\chi}$ or $\mathsf{iNTRU}_{\cR,Q,\chi,e}$, find a non-zero vector $(x_0,x_1)$ such that
	$\|(x_0,x_1)\|\leq \beta$ and $x_0+hx_1 = 0 \bmod Q$.
\end{definition}

\section{Recall the Gadget Trapdoors}\label{sec:recall}
While Ajtai's function $f_\matA$ is hard to invert for a random matrix $\matA$,
the inversion $f^{-1}_\matA$ can be easily computed with a short trapdoor.
The most famous and efficient lattice trapdoors are based on the lattice gadget framework developed in~\cite{micciancio2012trapdoors}.
In a gadget trapdoor scheme, the inversion of $f_\matA$ is converted into the
gadget inversion, i.e. the inversion of $f_\matG$ for a gadget matrix $\matG$.
The gadget inversion turns out to be highly simple and fast for some well-designed $\matG$. For better completeness and contrast, let us briefly recall the classical gadget trapdoor from~\cite{micciancio2012trapdoors} and its approximate variant from~\cite{chen2019approximate}.

\subsection{Exact gadget trapdoor from~\cite{micciancio2012trapdoors}}\label{subsec:MP12_gadget}
The earliest and most widely used gadget trapdoor is proposed by Micciancio and Peikert in~\cite{micciancio2012trapdoors}.
In the Micciancio-Peikert trapdoor, the gadget matrix is
$\matG = \matI_n \otimes \vecg^t\in\mathbb Z^{n\times m'}$ where $\vecg = (1, b,\cdots,b^{k-1}),k = \ceil{\log_b(Q)}$ and $m'=nk$.
The public matrix is
\[\matA = [\bar{\matA} \mid \matG - \bar{\matA}\matR] \in \Z_Q^{n\times m}\]
where $\bar{\matA}\in\mathbb Z_Q^{n\times\bar m},m=\bar m+m'$ and
$\matR$ is a secret matrix of small entries such that $\bar{\matA}\matR$ is either statistically near-uniform or computationally pseudorandom under certain assumptions. In this paper, we are interested in the pseudorandom case that offers better practicality due to the  smaller dimension of $\matA$.

Let $\matT = \left[\begin{smallmatrix}
	\matR \\
	\matI
\end{smallmatrix}\right]$, then $\matA\matT = \matG\bmod Q$.
This linear relation gives a direct transformation from $f^{-1}_\matA$ to the gadget inversion $f^{-1}_\matG$:
given a target $\vecu \in \Z_Q^n$, $\vecx = \matT\vecx'$ is a short preimage of $f^{-1}_\matA(\vecu)$ when
$\vecx'$ is a short preimage of $f^{-1}_\matG(\vecu)$.
Many applications, e.g. digital signatures, also need the preimage distribution to be simulatable without using the trapdoor
for \emph{uniformly random} targets for security purpose.
To this end, a common approach is to make the preimage distribution statistically close to some Gaussian independent of the trapdoor by
adding some perturbation following the idea of~\cite{peikert2010efficient}.
More concretely, the inversion $f^{-1}_\matA(\vecu)$ in the Micciancio-Peikert framework proceeds as follows:
\begin{enumerate}
	\item \emph{(Perturbation sampling)} Sample $\vecp$ from
	$D_{\Z^m, \sqrt{\Sigma_p}}$ where $\Sigma_p = s^2\matI_m - r^2\matT\matT^t$
	\item Compute $\vecu' = \vecu - \matA\vecp\bmod Q$
	\item \emph{(Gadget sampling)} Sample $\vecx'$ from
	$D_{\Lambda^{\perp}_{Q,\vecu'}(\matG), r}$
	\item Output the preimage $\vecx = \vecp + \matT\vecx' \bmod Q$
\end{enumerate}
The required parameter conditions by the Gaussian sampling include $r \geq \smootheps(\Lambda^{\perp}_Q(\matG))$ and $s \geq r\cdot s_1(\matT)$.

\subsection{Approximate gadget trapdoor from~\cite{chen2019approximate}}\label{subsec:CGM19_gadget}
In~\cite{chen2019approximate}, Chen, Genise and Mukherjee introduced the notion of \emph{approximate trapdoor}.
Such a trapdoor allows to approximately invert Ajtai's function $f_\matA$, i.e. to find a short preimage $\vecx$ of $f^{-1}_\matA(\vecu)$ such that $\matA\vecx = \vecu - \vece \mod Q$ for some short $\vece$. The vector $\vece$ is termed \emph{approximate error} or simply \emph{error}.
An approximate variant of the Micciancio-Peikert gadget trapdoor was given  in~\cite{chen2019approximate}.
In the Chen-Genise-Mukherjee trapdoor, the gadget matrix is $\matF = \matI_n \otimes \vecf^t\in\mathbb Z^{n\times m'}$ where $\vecf = (b^l, b^{l+1},\cdots,b^{k-1})$ is truncated from the gadget $\vecg$ in the exact case, and $m'=n(k-l)$. The public matrix accordingly becomes
\[\matA = [\bar{\matA} \mid \matF - \bar{\matA}\matR] \in \Z_Q^{n\times m},\]
where $m=\bar m+m'$.
Compared to the exact gadget, the approximate variant substantially reduces the dimension of $\matA$ and thus leads to more practical hash-and-sign signatures.

Let $\matT = \left[\begin{smallmatrix}
	\matR \\
	\matI
\end{smallmatrix}\right]$ and $\matD = \matI_n \otimes \vecd^t$ where $\vecd = (1, b,\cdots,b^{l-1})$. Then the exact gadget $\matG = \matI_n\otimes[\vecd^t \mid \vecf^t]$.
The approximate inversion follows the spirit of transforming $f^{-1}_\matA$ to the (approximate) gadget inversion. Given a target $\vecu$, it proceeds as follows:
\begin{enumerate}
	\item \emph{(Perturbation sampling)} Sample $\vecp$ from
	$D_{\Z^m, \sqrt{\Sigma_p}}$ where $\Sigma_p = s^2\matI_m - r^2\matT\matT^t$
	\item Compute $\vecu' = \vecu - \matA\vecp\bmod Q$
	\item \emph{(Gadget sampling)} Sample $\vecx'$ from
	$D_{\Lambda^{\perp}_{Q,\vecu'}(\matG), r}$
	\item \emph{(Preimage truncation)}
    Let $\vecx'=(\vecx'_1,\ldots,\vecx'_n)$ with $\vecx'_i\in\mathbb Z^k$. Set $\vecx''_i$ as the last $(k-l)$ entries of $\vecx'_i$ and $\vecx''=(\vecx''_1,\ldots,\vecx''_n)$
	\item Output the preimage $\vecx = \vecp + \matT\vecx'' \bmod Q$
\end{enumerate}
Let $\vecx' := (\vecx''',\vecx'')$, then the approximate error is
\[\vece = \vecu - \matA\vecx = \vecu - \matA\vecp - \matF\vecx'' = \vecu' - \matF\vecx'' = \matD\vecx'''\bmod Q. \]
For uniformly random $\vecu$, the distribution of $(\vecu,\vecx,\vece)$ can be simulated by sampling $\vecx \leftarrow D_{\Z^m,s}$ and $\vece \leftarrow D_{\Z^n,r\cdot\|\vecd\|}$ and then setting $\vecu = \matA\vecx+\vece\bmod Q$. The required parameter conditions include $r \geq \smootheps(\Lambda^{\perp}_Q(\matG))$ and $s \geq C\cdot r\cdot s_1(\matT)$ where $C$ is a small constant for commonly-used trapdoors.

\subsection{Equivalence between exact and approximate trapdoors}
Recall that the approximate trapdoor allows to sample a short preimage $\vecx$ such that $\matA\vecx+\vece=\vecu\bmod Q$ with a short error $\vece$.
When $\matA = [\matI_n \mid \matA']$, one can transform the  approximate preimage $\vecx = (\vecx_0,\vecx_1)$ and the error $\vece$
into an exact preimage $\vecx' = (\vecx_0 + \vece, \vecx_1)$ such that $\matA\vecx' = \vecu\bmod Q$.
Hence the exact and approximate trapdoors are somewhat equivalent from an algorithmic aspect.
This equivalence is characterized in the reduction form as follows.
\begin{lemma}[\cite{chen2019approximate}, Lemma 3.5, adapted]\label{lemma:equivalence}
	For $n,m,Q\in\Z,\alpha,\beta\geq 0$
	\begin{itemize}
		\item $\mathsf{HNF.ApproxISIS}_{n,m,Q,\alpha,\beta}\leq_p\mathsf{HNF.ISIS}_{n,m,Q,\beta}$ for any $\alpha\geq0$
		\item $\mathsf{HNF.ISIS}_{n,m,Q,\alpha+\beta}\leq_p\mathsf{HNF.ApproxISIS}_{n,m,Q,\alpha,\beta}$
	\end{itemize}
\end{lemma}
\begin{remark}
	Lemma~\ref{lemma:equivalence} simply takes $(\alpha+\beta)$ as the bound of the size of the exact preimage $\vecx' = (\vecx_0 + \vece, \vecx_1)$.
	When it comes to concrete security estimate, this additive bound is loose and a more accurate approach is to estimate
	$\|\vecx_0 + \vece\|$ and $\|\vecx_1\|$ separately. The term $\|\vecx_0 + \vece\|$ can be estimated based on the Pythagorean additive property when $\vecx_0$ and $\vece$ are Gaussian-like. Moerover, we consider the unbalanced sizes of $\vecx_0 + \vece$ and $\vecx_1$ in later security estimates.
\end{remark}


\section{Compact Gadget for Approximate Trapdoor}\label{sec:gadget}
We present a new gadget for approximate trapdoors in this section.
In contrast with existing gadgets from~\cite{micciancio2012trapdoors,chen2019approximate}, our gadget matrix is of size only \emph{$n$-by-$n$}, which allows more compact public keys and trapdoors.
At the core of our construction is a new type of approximate gadget sampler that we term \emph{semi-random} sampler.
In this sampler, the preimage is randomly sampled, whereas the error is \emph{deterministically fixed} by the target.
While the semi-random sampler loses some randomness of the error part, the distributions of the preimages and the errors can be still simulatable for \emph{uniformly random targets}. This suffices for the need of the application of hash-and-sign signatures.

\subsection{Description of our gadget trapdoor}
This section gives a general description of our gadget trapdoor and the semi-random sampler.
We believe that such a general description can guide further study of new gadget designs.

Let $\matP \in \Z^{n\times n}$ denote the gadget matrix used in our trapdoor construction and $\matQ \in \Z^{n\times n}$ such that
\[\matP\matQ = Q\cdot \matI_n.\]
The public matrix is $\matA\in\mathbb Z_Q^{n\times m}$ with $m > n$ and the approximate trapdoor for $\matA$ is defined as a matrix $\matT\in\mathbb Z^{m\times n}$ such that
\[\matA\matT=\matP\mod Q.\]
Then the approximate trapdoor inversion is transformed to the approximate gadget inversion implemented by our semi-random sampler.
\begin{remark}
	Our trapdoor can be instantiated under different assumptions:
	\begin{itemize}
		\item LWE-based: $\matA = [\matI \mid \bar{\matA} \mid \matP + \bar{\matA}\matS + \matE]$ and $\matT = [-\matE^t \mid -\matS^t \mid \matI]^t$;
		\item NTRU-based: $\matA = [\matI \mid (\matP-\matF)\cdot \matG^{-1}]$ and $\matT = [\matF^t \mid \matG^t]^t$.
	\end{itemize} 
	See Sections~\ref{sec:ntru} and~\ref{sec:lwe} for more details.
\end{remark}

Given a target $\vecu'$, the semi-random gadget sampler outputs a preimage $\vecx'$ such that $\matP\vecx' = \vecu' - \vece \bmod Q$ for some small error $\vece$.
It proceeds in two steps: 
(1) \emph{deterministic error decoding} and (2) \emph{random preimage sampling}. In the first step, the sampler computes the error $\vece$ such that $\vecu' - \vece = \matP\vecc\in \lat(\matP)$.
This can be done by lattice decoding algorithms, e.g. Babai's CVP algorithms~\cite{babai1986lovasz}. We denote by $\algDecode$ the deterministic lattice decoder and use it in a black-box way.
Typically, the output errors are identical for all vectors in a coset $\vect + \lat(\matP)$.
We denote by $E(\matP)$ the set of all possible errors and write $\vece = (\vecu' \bmod \lat(\matP))$ the error for $\vecu'$.
The next step is to sample the preimage $\vecx'$ from  $D_{\lat(\matQ)+\vecc,r}$. Let $\vecx' = \matQ\vecv+\vecc$ for $\vecv \in \Z^n$.
One can verify that
\begin{equation}\label{eq:gadget_correct}
\matP\vecx' = \matP\matQ\vecv + \matP\vecc = \vecu' - \vece \bmod Q.
\end{equation}
A formal description is given in Algorithm~\ref{alg:gadget}.

\begin{algorithm}[htbp]
	\caption{$\algGadget(\vecu', r, \matP, \matQ)$}{\label{alg:gadget}}
	\begin{algorithmic}[1]
		\REQUIRE matrices $\matP, \matQ \in \Z^{n\times n}$ such that $\matP\matQ = Q\cdot \matI_n$ and $r \geq \smootheps(\lat(\matQ))$
		\ENSURE a sample $\vecx' \sim D_{\Z^n,r}$
		conditioned on $\matP\vecx' = \vecu' - \vece \bmod Q$ and  $\vece \in E(\matP)$.
		\STATE $(\vecc,\vece)\leftarrow \algDecode(\vecu', \matP)$ such that $\vecc \in \Z^n$ and $\vecu' - \vece = \matP\vecc$
		\STATE $\vecx' \leftarrow D_{\lat(\matQ)+\vecc,r}$
		\RETURN $\vecx'$
	\end{algorithmic}
\end{algorithm}

The correctness of Algorithm~\ref{alg:gadget} is shown in Lemma~\ref{lem:gadgetcorrect}.
\begin{lemma}\label{lem:gadgetcorrect}
	Algorithm~\ref{alg:gadget} is correct. More precisely,
	let $\matP, \matQ \in \Z^{n\times n}$ such that $\matP\matQ = Q\cdot \matI_n$ and $r\geq \smootheps(\lat(\matQ))$.
	Then the output $\vecx'$ of $\algGadget(\vecu', r, \matP, \matQ)$ follows the distribution of $D_{\Z^n,r}$ conditioned on
	$\matP\vecx' = \vecu' - \vece \bmod Q$ with $\vece \in E(\matP)$.
\end{lemma}
\begin{proof}
	Given $\vecu'$, there exists a unique error $\vece = (\vecu' \bmod \lat(\matP))$ satisfying $\vece \in E(\matP)$ and $\vecu' - \vece = \matP\vecc\in \lat(\matP)$.
	For $\vecx'$ such that $\matP\vecx' = \vecu' - \vece \bmod Q$, let $\matP\vecx' = \vecu' - \vece + Q\vecv$, then
	$\matP\vecx' = \matP(\vecc+\matQ\vecv)$ and thus $\vecx' \in \lat(\matQ)+\vecc$.
	For $\vecx' \in \lat(\matQ)+\vecc$, let $\vecx' - \vecc = \matQ\vecv$ for some $\vecv \in \Z^n$, then $\matP\vecx' = \vecu' - \vece \bmod Q$ as shown by Eq.~\eqref{eq:gadget_correct}.
	Therefore $\matP\vecx' = \vecu' - \vece \bmod Q$ holds if and only if $\vecx' \in \lat(\matQ)+\vecc$. The proof is completed.
	\qed
\end{proof}

We now prove that for uniformly random $\vecu'$,  the preimage and error distributions of $\algGadget(\vecu', r, \matP, \matQ)$ can be simulated.
\begin{lemma}\label{lem:gadgetsimulation}
	Let $\matP, \matQ \in \Z^{n\times n}$ such that $\matP\matQ = Q\cdot \matI_n$ and $r\geq \smootheps(\lat(\matQ))$
	with some negligible $\epsilon > 0$.
	Let $\chi_\vece$ be the distribution of
	$(\vecv \bmod \lat(\matP)) \in E(\matP)$ where $\vecv \leftarrow U(\Z^n_Q)$.
	Then the following two distributions are statistically close.
	\begin{enumerate}\setlength\itemsep{0.1em}
		\item First sample $\vecu' \leftarrow U(\Z^n_Q)$, then sample $\vecx' \leftarrow \algGadget(\vecu', r, \matP, \matQ)$, compute $\vece = (\vecu' \bmod \lat(\matP))$, output $(\vecx', \vecu', \vece)$;
		\item First sample $\vece \leftarrow \chi_\vece$, then sample $\vecx' \leftarrow D_{\Z^n,r}$, set $\vecu' = \vece + \matP\vecx' \bmod Q$, output $(\vecx', \vecu', \vece)$.
	\end{enumerate}
\end{lemma}
\begin{proof}
	The supports of two distributions are identical as follows:
	$$\set{(\vecx', \vecu', \vece) \in \Z^n \times \Z^n_Q \times E(\matP) \mid \vecu' = \vece + \matP\vecx' \bmod Q}.$$
	Distribution 1 outputs $(\vecx', \vecu', \vece)$ with probability
	\[P_1[(\vecx', \vecu', \vece)] = \frac{1}{Q^n}P_1[\vecx' | \vecu'] = \frac{1}{Q^n}\cdot \frac{\rho_r(\vecx')}{\rho_r(\lat(\matQ)+\vecc)} \]
	and Distribution 2 with
	\[
	P_2[(\vecx', \vecu', \vece)] = \frac{1}{\det(\matP)}\cdot\frac{\rho_r(\vecx')}{\rho_r(\Z^n)} = \frac{\det(\matQ)}{Q^n}\cdot\frac{\rho_r(\vecx')}{\rho_r(\Z^n)}.
	\]
	Since $r\geq \smootheps(\lat(\matQ))$ and
	$\rho_r(\Z^n) = \sum_{\vecc \in \mathcal{P}(\matQ) \cap \Z^n} \rho_r(\vecc + \lat(\matQ))$,
	Lemma~\ref{lem:smooth} shows
	\[\rho_r(\lat(\matQ)+\vecc) \in  \left[\frac{1-\epsilon}{1+\epsilon}, \frac{1+\epsilon}{1-\epsilon}\right]\cdot\frac{\rho_r(\Z^n)}{\det(\matQ)}.\]
	Hence $P_1[(\vecx', \vecu', \vece)] \in \left[\frac{1-\epsilon}{1+\epsilon}, \frac{1+\epsilon}{1-\epsilon}\right]\cdot P_2[(\vecx', \vecu', \vece)]$ and we complete the proof.
	\qed
\end{proof}

Algorithm~\ref{alg:presamp} illustrates the approximate trapdoor inversion algorithm by using our gadget.
The output preimage $\vecx$ satisfies that
$$\matA\vecx = \matP\vecx' + \matA\vecp = \vecu' - \vece + \matA\vecp = \vecu - \vece \bmod Q.$$
Therefore the approximation error $\vece$ in $\algPreSamp(\matA,\matT,\vecu,r,s)$ is exactly the one in $\algGadget(\vecu', r, \matP, \matQ)$: for uniformly random $\vecu$, the error $\vece$ follows the distribution $\chi_\vece$ defined in Lemma~\ref{lem:gadgetsimulation}.

\begin{algorithm}[H]
	\caption{$\algPreSamp(\matA,\matT,\vecu,r,s)$}{\label{alg:presamp}}
	\begin{algorithmic}[1]
		\REQUIRE $(\matA,\matT) \in \Z_Q^{n\times m} \times \Z^{m \times n}$ such that $\matA\matT = \matP \bmod Q$, a vector $\vecu \in \Z_Q^n$, \\
		$r \geq \smootheps(\lat(\matQ))$ and $s^2\matI_m\succ r^2\matT\matT^t$
		\ENSURE an approximate preimage $\vecx$ of $\vecu$ for $\matA$.
		\STATE $\vecp \gets D_{\Z^m, \sqrt{\Sigma_p}}$ where $\Sigma_p = s^2\matI_m - r^2\matT\matT^t$
		\STATE $\vecu' = \vecu - \matA\vecp \bmod Q$
		\STATE $\vecx' \gets \algGadget(\vecu', r, \matP, \matQ)$
		\RETURN $\vecx = \vecp + \matT\vecx'$
	\end{algorithmic}
\end{algorithm}

Let $\mathbf L=[\mathbf I_m\mid\mathbf T]$. The next lemma characterizes the distribution of the linear transformation on the concatenation of $\mathbf p\leftarrow D_{\mathbb Z^m,\sqrt{\Sigma_{p}}}$ and $\vecx'\leftarrow D_{\mathbb Z^n,r}$, which represents the convolution step, i.e.,
$$
\mathbf x=\mathbf p+\mathbf T\vecx'=\mathbf{L}\cdot(\mathbf p,\vecx').
$$

\begin{lemma}\label{lem:lineartrans}
  Let $r\geq \smootheps(\Z^n)$. The distribution $\mathbf L\cdot D_{\mathbb Z^{m+n},\sqrt{\Sigma_p\oplus r^2\mathbf I_n}}$ is statistically close to $D_{\mathbb Z^m,s}$, if $s^2\geq\left(r^2+\smootheps(\mathbb Z^n)^2\right)\cdot \left(s_1(\matT)^2+1\right)$. 
\end{lemma}
\begin{proof}


  Let $\Lambda_{\mathbf L}=\mathbb Z^{m+n}\cap\textrm{ker}(\mathbf L)$ that is an integer lattice.
  By Theorem~\ref{thm:lineartransform}, it suffices to show $\sqrt{\Sigma_p\oplus r^2\mathbf I_n}\geq\eta_{\epsilon}\big(\Lambda_{\mathbf L}\big)$.
  Let $\matB=\left[
  \begin{smallmatrix}
  	\matT \\
  	-\matI_n
  \end{smallmatrix}
  \right]$, then $\matB$ is a basis of $\Lambda_{\mathbf L}$.
The dual basis of $\matB$ is
$$
\mathbf B^*=\mathbf B(\mathbf{B}^t\mathbf B)^{-1}=\left[
  \begin{array}{c}
    \mathbf T \\
    \mathbf{-I}_n
  \end{array}
\right]\left(\mathbf T^t\mathbf T+\mathbf I_n\right)^{-1}.
$$
According to the definition of smoothing parameter, we need to show
$$
\sqrt{\Sigma_p\oplus r^2\mathbf I_n}\geq\eta_{\epsilon}\big(\lat(\mathbf B)\big)
$$
i.e.,
\begin{equation*}
  (\mathbf B^*)^t(\Sigma_p\oplus r^2\cdot\mathbf I_n)\mathbf B^*\succ\eta^2_{\epsilon}(\mathbb Z^n).
\end{equation*}
This reduces to showing
\begin{equation*}
  \left(\mathbf T^t\mathbf T+\mathbf I_n\right)^{-t}\cdot\left(s^2\matT^t\matT-r^2(\matT^t\matT)^2+r^2\matI_n\right)\cdot\left(\mathbf T^t\mathbf T+\mathbf I_n\right)^{-1}\succ\eta^2_{\epsilon}(\mathbb Z^n).
\end{equation*}
Let $\mathbf T^t\mathbf T=\mathbf{UVU}^{-1}$ be the eigenvalue decomposition where $\matV=\mathrm{diag}(\lambda_1,\ldots,\lambda_n)$ with $\lambda_i$ being the eigenvalues. The left-hand side can be rewritten as
$$
\mathbf U(\mathbf V+\mathbf I_n)^{-t}\left(s^2\mathbf V- r^2\mathbf V^2+r^2\matI_n\right)(\mathbf V+\mathbf I_n)^{-1}\mathbf U^{-1},
$$
and we need to prove
$$
\frac{s^2\lambda_i- r^2\lambda_i^2+r^2}{(\lambda_i+1)^2}\geq\eta^2_{\epsilon}(\mathbb Z^n),
$$
i.e.
$$
s^2\geq(r^2+\smootheps(\mathbb Z^n)^2)\cdot \lambda_i+2\cdot\smootheps(\mathbb Z^n)^2+\frac{\smootheps(\mathbb Z^n)^2-r^2}{\lambda_i}.
$$
By some rountine computation, one can check that this condition is satisfied
when $r \geq \smootheps(\mathbb Z^n)$ and $\left(r^2+\smootheps(\mathbb Z^n)^2\right)\cdot \left(s_1(\matT)^2+1\right)$.
\qed
\end{proof}

We now prove that the preimage and error distributions are simulatable without knowing the trapdoor.
This property is needed in most trapdoor based use cases.
Our argument only holds for uniformly random $\vecu$ as in~\cite{chen2019approximate}.

\begin{theorem}\label{thm:maintheorem}
	Let $\matP, \matQ \in \Z^{n\times n}$ such that $\matP\matQ = Q\cdot \matI_n$. 
	Let $(\matA, \matT)$ be a matrix-approximate trapdoor pair,   
  $(r, s)$ satisfying 
  $s^2\geq\left(r^2+\smootheps(\mathbb Z^n)^2\right)\cdot \left(s_1(\matT)^2+1\right)$ and 
  $r\geq \smootheps(\lat(\matQ))$. 
	Then the following two distributions are statistically indistinguishable:
	\begin{equation*}
		\left\{
		(\matA, \vecx, \vecu, \vece):~\vecu \gets U(\Z_Q^n),~\vecx\gets \algPreSamp(\matA,\matT,\vecu,r,s),~\vece = \vecu-\matA\vecx \bmod Q
		\right\}
	\end{equation*}
	\begin{equation*}
		\left\{
		(\matA, \vecx, \vecu, \vece):~\vecx \gets D_{\Z^m,s},~\vece \gets \chi_{\vece},~\vecu = \matA\vecx + \vece \bmod Q
		\right\}.
	\end{equation*}
\end{theorem}
\begin{proof}
	Let  
	\begin{itemize}
		\item $\vecp \gets D_{\Z^m, \sqrt{\Sigma_p}}$ be a perturbation, 
		\item $\vecu\in \Z_Q^n$ be the target of $\algPreSamp(\matA,\matT,\vecu,r,s)$,
		\item $\vecu' = \vecu - \matA\vecp \bmod Q$ be the target of $\algGadget(\vecu', r, \matP, \matQ)$,
		\item $\chi_\vece$ be the distribution of
		$(\vecv \bmod \lat(\matP)) \in E(\matP)$ where $\vecv \leftarrow U(\Z^n_Q)$.
	\end{itemize}
	\noindent\textbf{Real distribution}: The real distribution of $(\matA, \vecx, \vecu, \vece)$ is
	{
	\begin{equation}
	\begin{aligned}
	&\matA, \vecu \gets U(\Z_Q^n), \vecp \gets D_{\Z^m, \sqrt{\Sigma_p}}, \vecu' = \vecu - \matA\vecp,
	\\
	&\vecx' \gets \algGadget(\vecu', r, \matP, \matQ), \vecx = \vecp + \matT\vecx', \vece = \vecu-\matA\vecx \bmod Q.	
	\end{aligned}
	\end{equation}
	}
	
	\noindent\textbf{Hybrid 1}: Instead of sampling $\vecu \gets U(\Z_Q^n)$, we sample $\vecu' \gets U(\Z_Q^n)$ and $\vecp \gets D_{\Z^m, \sqrt{\Sigma_p}}$, then compute
	$\vecu = \vecu' + \matA\vecp$. We keep $(\vecx', \vecx,\vece)$ unchanged. Clearly, the real distribution and Hybrid 1 are the same.

	\vspace{0.75em}
	\noindent\textbf{Hybrid 2}: Instead of sampling $\vecu', \vecx'$ and computing $\vece$ as in Hybrid 1, we sample $\vecx' \gets D_{\Z^{n},r}$ and $\vece\leftarrow\chi_{\vece}$, then compute $\vecu' = \matP\vecx' + \vece$.
	All other terms $(\vecp, \vecx, \vecu)$ remain unchanged. By Lemma~\ref{lem:gadgetsimulation}, Hybrid 1 and Hybrid 2 are statistically close.
	
	\vspace{0.75em}
	\noindent\textbf{Hybrid 3}: Instead of sampling $\vecp, \vecx'$ and computing $\vecx = \vecp + \matT\vecx'$ in Hybrid 2, we sample directly $\vecx \gets D_{\Z^m,s}$ and compute $\vecu = \matA\vecx + \vece \bmod Q$, where $\vece\gets\chi_{\vece}$ is as before. Note that in Hybrid 2,
$$
\vecu = \vecu' + \matA\vecp = \vece + \matP\vecx' + \matA\vecp = \vece + \matA(\matT\vecx' + \vecp) = \matA\vecx + \vece \bmod Q
$$
	and $\vecx = \vecp + \matT\vecx'$ follows the distribution $[\matI_m \mid \matT]\cdot D_{\Z^{m+n}, \sqrt{\Sigma_p \oplus r^2\matI_n}}$.
	By Lemma~\ref{lem:lineartrans}, Hybrid 3 and Hybrid 2 are statistically close. Now we complete the proof.  \qed
\end{proof}

\subsection{Simple instantiation and comparisons}\label{subsec:gadget_instance}
Our new gadget trapdoor has a very simple instantiation by using
\[(\matP,\matQ) = (p\matI_n,q\matI_n)\]
where $p, q \in \Z$ such that $Q = pq$. 
In this case, $\algDecode$ is implemented by coefficient-wise $\bmod~p$ operations and 
$E(\matP) = \Z_p^n$, $\chi_\vece = U(\Z_p^n)$.
Hence for uniformly random targets, the standard deviation of error coefficients is $\sqrt{\frac{p^2-1}{12}}$.
Since $\smootheps(\lat(\matQ)) = q\cdot \smootheps(\Z^n)$, the preimage size is about $\sqrt{n}\cdot q\cdot \smootheps(\Z^n)$.
Table~\ref{tab:gadget_comparison} shows the comparisons between previous gadgets and ours.

\begin{table}[H]\renewcommand\arraystretch{1.2}
	\centering
	\caption{Comparisons with the gadgets from~\cite{micciancio2012trapdoors} and~\cite{chen2019approximate}.
	Here $m'$ is the column number of the gadget matrix, $\vecx' \in \Z^{m'}$ is the preimage, $\vece\in \Z^n$ is the error and $\eta = \smootheps(\Z)$. }\label{tab:gadget_comparison}
	\setlength{\tabcolsep}{1.6mm}{}
	\begin{tabular}{p{1.5cm}<{\centering} p{2.8cm}<{\centering} p{1.4cm}<{\centering} p{1cm}<{\centering} p{2.2cm}<{\centering} p{1.6cm}<{\centering}}
		\toprule
		& Gadget & $Q$ & $m'$  & $\|\vecx'\|/\sqrt{m'}$ & $\|\vece\|/\sqrt{n}$ \\
		\hline
		\vspace{-.2cm}\cite{micciancio2012trapdoors}  & $\matI_n \otimes \vecg^t$, $\vecg = (1, b,\cdots,b^{k-1})$ & \vspace{-.2cm}$(b^{k-1},b^k]$ & \vspace{-.2cm}$nk$ & \vspace{-.2cm} $\approx \sqrt{(b^2+1)}\eta$ & \vspace{-.2cm} $0$ \\
		\vspace{-.2cm}\cite{chen2019approximate}  & $\matI_n \otimes \vecf^t$, $\vecf = (b^l, \cdots,b^{k-1})$ & \vspace{-.2cm}$(b^{k-1},b^k]$ & \vspace{-.2cm}$n(k-l)$ & $\approx \sqrt{(b^2+1)}\eta$  & $\approx b^l\eta$ \\
		This work  & $p\cdot\matI_n$ & $pq$ & $n$ & $\approx q\eta$ & $\approx\sqrt{\frac{p^2-1}{12}}$ \\
		\bottomrule
	\end{tabular}
\end{table}

The above instantiation of our approximate gadget offers significant advantages in terms of compactness, efficiency and parameter selection:

\paragraph{Compactness.}
Our gadget vector consists of only one entry, i.e. $p$,
while to the best of our knowledge, the gadget  from~\cite{chen2019approximate} requires at least three entries in practical applications.
As a direct consequence, the hash-and-sign signatures based on our gadget have much shorter key and signatures.
The reduced trapdoor size also results in a smaller Gaussian width and thus supports a smaller modulus.

\paragraph{Efficiency.}
Due to the semi-random sampler, the error in this simple instantiation is deterministically generated by modulo, which is highly efficient in terms of speed and randomness.
The preimage sampling boils down to only $n$ times sampling of $D_{q\Z+c,r}$.
By contrast, although \cite{chen2019approximate} proposed to replace the gadget $\vecg=(1,b,\cdots,b^k)$ with a truncated version $(b^{l},\cdots,b^{k})$, the gadget sampling is still performed over the gadget lattice defined by $\vecg$ and thus requires $k$ times integer sampling.
In addition, our gadget allows smaller trapdoors, which also reduces the cost of perturbation sampling.

\paragraph{Parameter selection.}
As mentioned before, the modulus associated with our gadget is $Q = pq$ and
the preimage and error sizes are linear in $q$ and $p$ respectively.
This is convenient for flexible and tight parameter choices.
However, for the gadget in \cite{chen2019approximate}, its error size is roughly proportional to $b^{l-1}$
and such an exponential growth heavily limits optimal parameter selection.

\vspace{.25cm}

In the rest of the paper, we will use the above simple gadget instantiation to build hash-and-sign signatures.
Nevertheless, the design space can be further expanded by taking some lattices with efficient decoding or sampling\footnote{Such remarkable lattices are listed in~\cite{ducas2022lattice}.} into account. 
More practical instantiations definitely need much efforts.
We leave this to future works.

\section{Efficient Hash-and-Sign Signatures over NTRU Lattices}\label{sec:ntru}
This section presents an NTRU-based hash-and-sign signature scheme, named \Robin, that is instantiated with the compact gadget in Section~\ref{sec:gadget}.
\Robin{} achieves good performance comparable to Falcon~\cite{falcon} and its variant Mitaka~\cite{espitau2022simpler}.
It also offers significant advantages from an implementation standpoint.
Its signing procedure is considerably simpler and easier to implement without floating point arithmetic.
Its secret key is one short vector instead of one short basis as in Falcon and Mitaka, which dramatically simplifies and accelerates the trapdoor generation.
\Robin{} can therefore be an attractive choice particularly in constrained environments.

\subsection{Description of the \Robin{} signature scheme}\label{subsec:ntru_alg}
\subsubsection{Parameters.}
The underlying NTRU is parameterized by the ring $\cR = \cR_n^- = \Z[x] / (x^n-1)$ and the modulus $Q$.
Let
\[\mathcal{T}(n, a, b) = \left\{ v \in \cR \Big|
\begin{aligned}
	& a\  \text{coefficients equal to }  1;\\
	v \text{ has exactly } & b\  \text{coefficients equal to } -1; \\
	&n-a-b \ \text{coefficients equal to } 0.
\end{aligned}\right\}.\]
The secret key $(f,g)$ are uniformly sampled from $\mathcal{T}(n, a, b)$.
The gadget matrix is $\matP = p\matI_n$ and the associated $\matQ = q\matI_n$ such that $pq = Q$.
Let $\alpha$ be the parameter controlling the quality of the trapdoor such that
$$\sqrt{s_1\left(\mathcal{M}(f\bar f + g\bar g)\right)} \leq \alpha \|(f,g)\| = \alpha \sqrt{2(a+b)}.$$
Let $\bar r = \smootheps(\Z^n)$ and $r \geq q\bar{r}$ be the width for the approximate gadget sampler.
Let $s \geq \frac{\sqrt{1+q^2}}{q}r\alpha \sqrt{2(a+b)}$ be the width for approximate preimages.
Let $\beta$ be the acceptance bound of $\|(z_0+e, \gamma z_1)\|$ where $(z_0,z_1)$ is the approximate preimage, $e$ is the approximate error and $\gamma = \frac{\sqrt{s^2 + (p^2-1)/12}}{s}$ such that $\|z_0+e\|\approx \gamma \|z_1\|$.

\subsubsection{Key generation.} The key generation of \Robin{} is very different from that of other NTRU-based hash-and-sign signatures Falcon and Mitaka.
Instead, it is similar to that of BLISS~\cite{ducas2013lattice} which is an NTRU-based Fiat-Shamir signature scheme.
More concretely, \Robin{} uses an inhomogeneous NTRU key pair in which
the secret key is composed of two short polynomials $(f, g)$ and the public key is $h = (p - g) / f \bmod Q$, then $hf+g = p \bmod Q$.
In addition, we partially apply the techniques suggested in~\cite{espitau2022simpler} to get a high-quality trapdoor in a short time.
The whole key generation is formally described in Algorithm~\ref{alg:ntru_keygen}.

\begin{algorithm}[htb]
	\caption{$\algNtruKG$}{\label{alg:ntru_keygen}}
	\begin{algorithmic}[1]
		\REQUIRE
		the ring $\cR = \cR_n^-$ with $n$ a prime, $Q = pq$, $(a, b) \in\Z^2$ and $\alpha > 0$.
		\ENSURE
		public key $h \in \cR / (Q\cdot\cR)$,
		secret key $(f,g) \in \cR^2$
		\STATE $f_1,\cdots, f_{K}, g_1,\cdots,g_{K} \overset{\$}{\leftarrow} \mathcal{T}(n, a,b)$ with $K=5$
		\FOR{$i = 1$ to $K$}
		\FOR{$j = 1$ to $K$}
		\STATE find $k \in \Z_n^*$ minimizing $s_1\left( \mathcal{M}\left(f_i\bar f_i + \sigma_k(g_j)\overline{\sigma_k(g_j)}\right) \right)$
		\STATE $(f,g) \gets (f_i,\sigma_k(g_j))$
		\IF {$\sqrt{s_1\left(\mathcal{M}(f\bar f + g\bar g)\right)} \leq \alpha \sqrt{2(a+b)}$}
		\STATE $h \leftarrow (p - g) / f \bmod Q$
		\STATE return $(h, (f,g))$
		\ENDIF
		\ENDFOR
		\ENDFOR
		\STATE restart
	\end{algorithmic}
\end{algorithm}

\subsubsection{Signing procedure.} Algorithm~\ref{alg:ntru_sign} shows the signing procedure that is in essence the approximate preimage sampling (Algorithm~\ref{alg:presamp}).
Given the hashed message $u$, Algorithm~\ref{alg:ntru_sign} samples a preimage $(z_0,z_1)$ such that $z_0 + hz_1 = u - e \bmod Q$ for small $e$.
Only $z_1$ is used as the actual signature, as the short term $(z_0+e) = u - hz_1 \bmod Q$ can be recovered during verification. 
We set the acceptance bound $\beta=1.04\cdot \Exp[\|(z_0+e,\gamma z_1)\|]$. 
We experimentally verified that the restart happens with probability $\approx 1\%$ for this setting. 

\begin{algorithm}[htb]
	\caption{$\algNtruSign$}{\label{alg:ntru_sign}}
	\begin{algorithmic}[1]
		\REQUIRE
		a message $\msg$, the NTRU key pair $(h, (f,g))$, $r \geq q\smootheps(\Z^n)$, \\
		$s \geq r\alpha \sqrt{2(a+b)}$, $\gamma = \frac{\sqrt{s^2 + (p^2-1)/12}}{s}$, $\beta > 0$.
		\ENSURE
		a signature $(\salt, z)$
		\STATE $\matA \gets [\matI_n \mid \mathcal{M}(h)], \matT \gets \begin{bmatrix}
			\mathcal{M}(g) \\ \mathcal{M}(f)
		\end{bmatrix}$
		\STATE $\salt \overset{\$}{\leftarrow} \set{0,1}^{320}$, $u \gets \hash(\msg,\salt)$
		\STATE $(z_0, z_1) \gets \algPreSamp(\matA,\matT,u,r,s)$
		\STATE $e \gets u - (z_0+z_1h) \bmod Q$
		\IF{$\|(z_0+e, \gamma z_1)\| > \beta$}
		\STATE restart
		\ENDIF
		\STATE return $(\salt, z_1)$
	\end{algorithmic}
\end{algorithm}

\subsubsection{Verification.} The preimage $(z_0+e,z_1)$ is short and $(z_0+e) + hz_1 = u \bmod Q$.
The verification is to check the shortness of $(u-hz_1, z_1)$.
To balance the sizes of $u-hz_1=z_0+e$ and $z_1$, we scale $z_1$ by a factor $\gamma = \frac{\sqrt{s^2 + (p^2-1)/12}}{s}$ in the shortness check.
A formal description is given in Algorithm~\ref{alg:ntru_verify}.

\begin{algorithm}[htb]
	\caption{$\algNtruVerify$}{\label{alg:ntru_verify}}
	\begin{algorithmic}[1]
		\REQUIRE
		a signature $(\salt, z)$ of a message $\msg$, the public key $h$,\\
		$\gamma = \frac{\sqrt{s^2 + (p^2-1)/12}}{s}$, $\beta > 0$.
		\ENSURE
		Accept or Reject
		\STATE $u \gets \hash(\msg,\salt)$, $z' \gets (u - hz)\bmod Q$
		\STATE Accept if $\|(z', \gamma z)\| \leq \beta$, otherwise Reject
	\end{algorithmic}
\end{algorithm}

\subsection{Security analysis}\label{subsec:ntru_security}
We now give a security proof for \Robin.
To start with, we need some treatment on the scaling factor $\gamma$ which modified the shortness condition for better concrete security.
To this end, we introduce a variant of NTRU-SIS in the twisted norm as follows.
\begin{definition}[NTRU-SIS in the twisted norm, $\mathsf{NTRUSIS}^{\|\cdot\|_\gamma}_{\cR,Q,\chi,\beta}$]
	Let $\cR=\Z[x]/(x^n - 1)$ with $n$ a prime.
	Let $Q > 0$ be an integer, $\chi$ be a distribution over $\cR$ and $\beta > 0$, $\gamma \geq 1$.
	Given a random NTRU public key $h$ of either $\mathsf{NTRU}_{\cR,Q,\chi}$ or $\mathsf{iNTRU}_{\cR,Q,\chi,e}$, find a non-zero vector $(x_0,x_1)$ such that
	$\|(x_0,\gamma x_1)\|\leq \beta$ and $x_0+hx_1 = 0 \bmod Q$.
\end{definition}
It is easy to verify that
\[
\mathsf{NTRUSIS}_{\cR,Q,\chi,\beta\gamma} \leq_p \mathsf{NTRUSIS}^{\|\cdot\|_\gamma}_{\cR,Q,\chi,\beta} \leq_p \mathsf{NTRUSIS}_{\cR,Q,\chi,\beta/\gamma},
\]
which shows the equivalence between NTRU-SIS and its twisted-norm version.

To prove the strong EU-CMA security of \Robin, we follow the same arguments for the GPV signatures~\cite{gentry2008trapdoors} and
combine Theorem~\ref{thm:maintheorem} showing that the preimage and error output by $\algPreSamp$ is simulatable for uniformly random targets. 
\begin{theorem}\label{thm:Robinsecurity}
	The \Robin{} signature scheme is strongly existentially unforgeable under a chosen-message attack in the random oracle model assuming the hardness of $\mathsf{NTRUSIS}^{\|\cdot\|_\gamma}_{\cR,Q,\chi,2\beta}$.
\end{theorem}
\begin{proof}
  Suppose, for contradiction, that there is an adversary $\mathcal A$ that breaks the strong EU-CMA security of \Robin{} with non-negligible probability $\varepsilon$.
  We construct a polynomial time algorithm $\mathcal S$ that solves $\mathsf{NTRUSIS}^{\|\cdot\|_\gamma}_{\cR,Q,\chi,2\beta}$ with probability close to $\varepsilon$.
  Given a random NTRU public key $h$, $\mathcal S$ runs
  $\mathcal A$ and simulates the random oracle $\mathsf H$ and signing oracle as follows:

  \begin{itemize}
    \item for the query to $\mathsf H$ on $(\mathsf{salt,msg})$,
    if $\mathsf H(\mathsf{salt,msg})$ is not queried, then
    $\mathcal S$ samples $(\vecz=(z_0,z_1), e)\gets D_{\mathbb Z^{2n},s} \times U(\mathbb Z_p^n)$,
    returns $u = z_0+z_1h+e\mod Q$ as the random oracle response and stores
    $((\mathsf{salt,msg}),z_0,z_1,e,u)$.
    Otherwise $\mathcal S$ looks up $((\mathsf{salt,msg}),z_0,z_1,e,u)$ and returns $u$ to $\mathcal A$.

	\item for every signing query on $\mathsf{msg}$, $\mathcal S$ samples $\mathsf{salt}\overset{\$}{\leftarrow}\{0,1\}^{320}$,
	$(\vecz=(z_0,z_1), e)\gets D_{\mathbb Z^{2n},s} \times U(\mathbb Z_p^n)$, then outputs $(\mathsf{salt},z_1)$ to $\mathcal A$ as the signature, and stores $((\mathsf{salt,msg}),z_0,z_1,e,u = z_0+z_1h+e\mod Q)$ in the random oracle storage.
  \end{itemize}

Without loss of generality, assume that before outputting the signature forgery $(\mathsf{salt^*},z_1)$ for the message $\mathsf{msg^*}$, $\mathcal A$ queries $\mathsf H$ on $(\mathsf{salt}^*, \mathsf{msg^*})$. Then $\mathcal S$ computes $z'_0=\mathsf{H(salt^*,msg^*)}-z_1h\mod Q$ and looks up $((\mathsf{salt^*,msg^*}),z^*_0,z^*_1, e^*,u^*)$ in its local storage. Finally, $\mathcal S$ outputs $(z^*_0+e^*-z'_0,z^*_1-z_1)$ as a solution.

By Theorem~\ref{thm:maintheorem}, the view of $\mathcal A$ in the real scheme is indistinguishable from the view provided by $\mathcal S$ except with negligible probability $Q_{sign}^2/2^{320}$, in which case repeated signature queries on the same message $\mathsf{msg}$ use the same $\mathsf{salt}$. It remains to prove that $z^*_1\neq z_1$. In fact, if $\mathsf{msg}^*$ has been queried to the signing oracle before, then the above inequality holds by the definition of a successful forgery; if $\mathsf{msg}^*$ has not been queried to the signing oracle, then $z_1$ is with high min-entropy for appropriate parameters, so $z^*_1\neq z_1$ with overwhelming probability. \qed
\end{proof}

\subsection{Concrete parameters}\label{subsec:ntru_param}
We provide 3 parameter sets for \Robin{} in Table~\ref{tab:robin_param} for the NIST security levels 1, 3 and 5 respectively.
In all parameter sets, $Q$ is a power of $2$ and $n$ is a prime such that the order of $2$ in $\Z_n$ is either $n-1$ or $\frac{n-1}{2}$ as suggested in~\cite{hoffstein2017choosing}. 
Let $b = a - 1 = \floor{\frac n 4}$.
We choose $\alpha \approx 1.7$ to guarantte the key generation terminate with a small number of trials.
The parameter $\bar r = \smootheps(\Z^n)$ uses $\epsilon = 2^{-36}$ that suffices to ensure a security level $\leq 256$ bits with up to $2^{64}$ signature queries as per~\cite{prest2017sharper}.

In Table~\ref{tab:robin_param}, the numbers of signature sizes are made according to the entropy of the preimage.
This can be efficiently obtained by using batch encoding with ANS (Asymetric Numeral System) as in~\cite{espitau2022shorter}.
The concrete security is estimated by the usual cryptanalytic methods for lattice-based cryptography. Details are provided in Supplementary Material~\ref{appendix:securityestimate}.

\begin{table}[H]\renewcommand\arraystretch{1.2}
	\centering\begin{threeparttable}
		\begin{tabular}{|p{4.2cm}<{\centering}| p{2.5cm}<{\centering} | p{2.5cm}<{\centering} | p{2.5cm}<{\centering} |}
			\hline
			& $\Robin$-701 & $\Robin$-1061 & $\Robin$-1279\\
			\hline
			{Target security level} & NIST-I & NIST-III & NIST-V \\
			{$n$} & $701$ & $1061$ & $1279$ \\
			{$(Q,p,q)$} & $(16384, 2048, 8)$ & $(32768, 4096, 8)$ & $(32768, 4096, 8)$ \\
			{$(a,b)$} & $(176,175)$ & $(266,265)$ & $(320,319)$ \\
			{$\alpha$} & $1.65$ & $1.7$ & $1.75$ \\
			{$r$} & $10.22$ & $10.28$ & $10.31$ \\
			{$s$} & $449.8$ & $573.8$ & $650.4$ \\
			{$\gamma$} & $1.65$ & $2.29$ & $2.07$ \\
			{$\beta$} & $28928.7$ & $62965.5$ & $70983.7$ \\
			\hline
			{Public key size (in bytes)} & $1227$ & $1990$ & $2399$ \\
			{Signature size (in bytes)} & $992$ & $1527$ & $1862$ \\
			\hline
			{Key recovery security (C/Q)}  & 116 / 105  & 181 / 165 & 228 / 207  \\
			\hline
			{Forgery security (C/Q)} &  130 / 118 & 214 / 195 & 264 / 240  \\
			\hline
		\end{tabular}
		\vspace{.2cm}
		\caption{Suggested parameters for $\Robin$. }\label{tab:robin_param}
	\end{threeparttable}
\end{table}

\subsection{Comparison with Falcon and Mitaka}\label{subsec:ntru_comparison}
\subsubsection{Implementation.}
\Robin{} has significant advantages from the implementation standpoint.
First, \Robin{} uses only one vector as the NTRU secret and avoids the notoriously complex NTRU trapdoor generation.
This can be crucial to the implementations and the key storage, especially when the key management for the entire lifecycle
is required (e.g. by the FIPS 140-2~\cite{FIPS}).
Second, \Robin{} has an online/offline structure as Mitaka, and its online operations are simple and fully over integers, which surpasses Mitaka.
In particular, base samplings in the online phase are in the form  $D_{q\Z+c,r}$ with $c\in\Z$.
This is beneficial for further optimization and side-channel protections.
Third, the offline sampling can also be implemented without resorting floating-point numbers by the technique in~\cite{ducas2020integral}.
The integral implemenation seems more convenient compared to the integer version of Mitaka.

\subsubsection{Performance.}
The size of \Robin{} is comparable to that of Falcon and Mitaka:
the total bandwidth (i.e. public key size + signature size) of \Robin{} is larger by $\approx 40\%$ than that of Falcon and 
by $25\% - 35\%$ than that of Mitaka.
Detailed comparisons are shown in Table~\ref{tab:NTRU_comparison}.

\begin{table}[h]\renewcommand\arraystretch{1.2}
	\centering
	\caption{Comparisons in terms of sizes with Falcon and Mitaka.
		For a fair comparison, all signature sizes are estimated as per the entropic bound, which can be closely obtained by entropic encoding as shown  in~\cite{espitau2022shorter}.
	}\label{tab:NTRU_comparison}
	\setlength{\tabcolsep}{2mm}{}
	\begin{threeparttable}
	\begin{tabular}{rcccc}
		\toprule
		& Security level & Pub. Key size (in bytes) & Sig. size (in bytes) \\
		\hline
		Falcon-512 & NIST-I & 896 & 643 \\
		Mitaka-648 & NIST-I & 972 & 807 \\
		{\bf \Robin-701} & NIST-I & {\bf 1227} & {\bf 992} \\
		\hline
		Mitaka-864 & NIST-III & 1512 & 1148 \\
		{\bf \Robin-1061} & NIST-III & {\bf 1990} & {\bf 1527}  \\
		\hline
		Falcon-1024 & NIST-V & 1792 & 1249 \\
		Mitaka-1024 & NIST-V & 1792 & 1376 \\
		{\bf \Robin-1279} & NIST-V & {\bf 2399} & {\bf 1862} \\
		\bottomrule
	\end{tabular}
\end{threeparttable}
\end{table}

\section{Shorter LWE-based Hash-and-Sign Signatures}\label{sec:lwe}
The LWE-based hash-and-sign signatures are rarely seen as a competitive post-quantum candidate in contrast to their NTRU and Fiat-Shamir counterparts, mainly due to their large sizes.
In this section, we fill the gap in practical LWE-based hash-and-sign signatures with a new scheme \Eagle.
\Eagle{} is instantiated with our compact gadget and based on Ring-LWE.
It achieves a desirable performance: \Eagle{} is substantially smaller than the state-of-the-art LWE-based hash-and-sign signatures~\cite{chen2019approximate}, and
even smaller than the LWE-based Fiat-Shamir signature scheme Dilithium~\cite{dilithium}.
While \Eagle{} is less efficient than the NTRU-based instantiation~\Robin{}, we believe it is of practical interest given the preference for using LWE to NTRU sometimes.

\subsection{Description of the \Eagle{} signature scheme}\label{subsec:rlwe_alg}
\subsubsection{Parameters.}
\Eagle{} is based on the Ring-LWE assumption over $\cR = \cR_n^+ = \Z[x] / (x^n+1)$ with $n$ a power of $2$ and the modulus $Q = pq$.
\Eagle{} uses the secret with a fixed hamming weight in $\mathcal{T}(n, a, b)$ (defined in Section~\ref{subsec:ntru_alg}),
and let $\alpha$ be the parameter controlling the quality of the trapdoor such that
$\sqrt{s_1\left(\mathcal{M}(f\bar f + g\bar g)\right)} \leq \alpha \|(f,g)\| = \alpha \sqrt{2(a+b)}.$
Let $\bar r = \smootheps(\Z^n)$ and $r \geq q\bar r$ be the width for the approximate gadget sampler.
Let $s \geq \frac{\sqrt{1+q^2}}{q}r\alpha \sqrt{2(a+b)}$ be the width for approximate preimages.
Let $\beta$ be the acceptance bound of $\|(z_0+e, \gamma z_1, \gamma z_2)\|$ where $(z_0,z_1,z_2)$ is the approximate preimage, $e$ is the approximate error and $\gamma = \frac{\sqrt{s^2 + (p^2-1)/12}}{s}$ such that $\|z_0+e\|\approx \gamma \|z_1\|\approx \gamma \|z_2\|$.

\subsubsection{Key generation.}
The public key is essentially $(a,b=p-(af+g)\bmod Q)$ where $a$ is uniformly random over $\cR_Q = \cR / (Q\cdot\cR)$ and $f,g\overset{\$}{\leftarrow} \mathcal{T}(n,a,b)$.
The polynomial $a$ is stored as a seed (of length 32 bytes), which halves the public key size.
We apply the techniques in~\cite{espitau2022simpler} as in \Robin{} to refine the quality of $(f,g)$. A formal description of the key generation is given in  Algorithm~\ref{alg:lwe_keygen}.

\begin{algorithm}[htb]
	\caption{$\algLWEKG$}{\label{alg:lwe_keygen}}
	\begin{algorithmic}[1]
		\REQUIRE
		the ring $\cR = \cR_n^+$ with $n$ a power of 2, $Q = pq$, and $\alpha > 0$.
		\ENSURE
		public key $(\seed_a,b)$,
		secret key $(f,g) \in \cR^2$
		\STATE $\seed_a \overset{\$}{\leftarrow} \set{0,1}^{256}$, $a \leftarrow \expand(\seed_a)$ \hfill\COMMENT{$\expand$ maps a seed to an element in $\cR$}
		\STATE $f_1,\cdots, f_{K}, g_1,\cdots,g_{K} \overset{\$}{\leftarrow} \mathcal{T}(n, a,b)$ with $K=5$
		\FOR{$i = 1$ to $K$}
		\FOR{$j = 1$ to $K$}
		\STATE find $k \in \Z_n^*$ minimizing $s_1\left( \mathcal{M}\left(f_i\bar f_i + \sigma_k(g_j)\overline{\sigma_k(g_j)}\right) \right)$
		\STATE $(f,g) \gets (f_i,\sigma_k(g_j))$
		\IF {$\sqrt{s_1\left(\mathcal{M}(f\bar f + g\bar g)\right)} \leq \alpha \sqrt{2(a+b)}$}
		\STATE $b \leftarrow p - (af+g) \bmod Q$
		\STATE return $((\seed_a,b), (f,g))$
		\ENDIF
		\ENDFOR
		\ENDFOR
		\STATE restart
	\end{algorithmic}
\end{algorithm}

\subsubsection{Signing procedure.}
Given the hashed message $u$, the signing procedure shown in Algorithm~\ref{alg:lwe_sign} samples a short preimage $(z_0,z_1,z_1)$ such that $z_0 + az_1 + bz_1 = u - e \bmod Q$ for a small $e$.
Only $(z_1,z_2)$ is used as the actual signature, as the short term $(z_0+e) = u - az_1-bz_2 \bmod Q$ can be recovered during verification. 
Again, the acceptance bound $\beta=1.04\cdot \Exp[\|(z_0+e,\gamma z_1)\|]$, which makes the restart happen with low probability. 

\subsubsection{Verification.}
The preimage $(z_0+e,z_1,z_2)$ is short and $(z_0+e) + az_1 + bz_2 = u \bmod Q$.
The verification is to check the shortness of $(u-az_1-bz_2, z_1, z_2)$.
A formal description is given in Algorithm~\ref{alg:lwe_verify}.

\begin{algorithm}[htb]
	\caption{$\algLWESign$}{\label{alg:lwe_sign}}
	\begin{algorithmic}[1]
		\REQUIRE
		a message $\msg$ and the key pair $((\seed_a,b), (f,g))$
		\ENSURE
		a signature $(\salt, (z_1,z_2))$
		\STATE $a \leftarrow \expand(\seed_a)$, $\salt \overset{\$}{\leftarrow} \set{0,1}^{320}$, $u \gets \hash(\msg,\salt)$
		\STATE $\matA \gets [\matI_n \mid \mathcal M(a) \mid \mathcal M(b)], \matT \gets
		\begin{bmatrix}
			\mathcal M(g) \\ \mathcal M(f) \\ \matI_n
		\end{bmatrix}$
		\STATE $(z_0, z_1, z_2) \gets \algPreSamp\left(\matA,\matT,u,r,s\right)$
		
		\STATE $e \gets u - (z_0+az_1+bz_2) \bmod Q$
		\IF{$\|(z_0+e, \gamma z_1, \gamma z_2)\| > \beta$}
		\STATE restart
		\ENDIF
		\STATE return $(\salt, (z_1,z_2))$
	\end{algorithmic}
\end{algorithm}
\begin{algorithm}[htb]
	\caption{$\algLWEVerify$}{\label{alg:lwe_verify}}
	\begin{algorithmic}[1]
		\REQUIRE
		a signature $(\salt, (z_1,z_2))$ of a message $\msg$, the public key $(\seed_a,b)$,\\
		$\gamma = \frac{\sqrt{s^2 + (p^2-1)/12}}{s}$, $\beta > 0$.
		\ENSURE
		Accept or Reject
		\STATE $u \gets \hash(\msg,\salt)$, $a \gets \expand(\seed_a)$, $z' \gets (u - az_1-bz_2) \bmod Q$
		\STATE Accept if $\|(z', \gamma z_1, \gamma z_2)\| \leq \beta$, otherwise Reject
	\end{algorithmic}
\end{algorithm}


\subsection{Security analysis}\label{subsec:lwe_security}
Similar to \Robin, the security of \Eagle{} is based on a variant of Ring-SIS in the twisted norm.
\begin{definition}[Ring-SIS in the twisted norm, $\mathsf{RSIS}_{\cR,m,Q,\beta}^{\|\cdot\|_{\gamma}}$]
	Let $\cR = \Z[x]/(x^n+1)$ with $n$ a power of $2$.
	Let $m,Q>0$ be integers and $\beta > 0$.
	Let $\cR_Q = \cR / (Q\cdot\cR)$.
	Given a uniformly random $\matA \in \cR_Q^{m}$, find a non-zero $\vecx=(x_0,\vecx_1) \in \cR\times \cR^{m-1}$ such that $\matA\vecx =0\bmod Q$ and $\|(x_0,\gamma\vecx_1)\| \leq \beta$.
\end{definition}

Theorem~\ref{thm:Eaglesecurity} shows the strong EU-CMA security of \Eagle. We omit the proof, as it follows the same argument with that of Theorem~\ref{thm:Robinsecurity}.
\begin{theorem}\label{thm:Eaglesecurity}
	The \Eagle{} signature scheme is strongly existentially unforgeable under a chosen-message attack in the random oracle model assuming the hardness of $\mathsf{RSIS}_{\cR,m,Q,\beta}^{\|\cdot\|_{\gamma}}$ and $\mathsf{RLWE}_{\cR,1,Q,\chi}$ with $\chi=U(\mathcal{T}(n, a,b))$.
\end{theorem}

\subsection{Concrete parameters}\label{subsec:lwe_param}
We provide 2 parameter sets for \Eagle{} in Table~\ref{tab:lwe_param}.
The public key size is computed as $n\cdot\log_2(Q)/8 + 32$ and the signature size is estimated as the entropic bound of the preimage plus 40 bytes for the salt. The details of concrete security estimate is shown in  Supplementary Material~\ref{appendix:securityestimate}.

\begin{table}[h]\renewcommand\arraystretch{1.2}
	\centering\begin{threeparttable}
		\begin{tabular}{|p{4.2cm}<{\centering} | p{2.5cm}<{\centering} | p{2.5cm}<{\centering} |}
			\hline
			 & \Eagle-512 & \Eagle-1024 \\
			\hline
			{Target security level} & 80-bit & NIST-III \\
			{$n$} & $512$ & $1024$ \\
			{$(Q,p,q)$} & $(16000, 2000, 8)$ & $(32400, 2700, 12)$ \\
			{$(a,b)$} & $(128,128)$ & $(256,256)$ \\
			{$\alpha$} & $1.7$ & $1.7$ \\
			{$r$} & $10.17$ & $15.42$ \\
			{$s$} & $394.2$ & $841.5$ \\
			{$\gamma$} & $1.36$ & $1.19$ \\
			{$\beta$} & $28493.5$ & $66118.5$ \\
			\hline
			{Public key size (in bytes)} & $928$ & $1952$ \\
			{Signature size (in bytes)} & $1406$ & $3052$ \\
			\hline
			{Key recovery security (C/Q)}  & 79 / 71  & 176 / 160  \\
			\hline
			{Forgery security (C/Q)} &   83 / 75 & 189 / 172  \\
			\hline
		\end{tabular}
		\vspace{.2cm}
		\caption{Suggested parameters for \Eagle. }\label{tab:lwe_param}
	\end{threeparttable}
\end{table}


\subsection{Comparison with LWE-based signatures}
Thanks to the compact gadget, \Eagle{} has much better compactness than existing LWE-based hash-and-sign signatures.
We first compare \Eagle{} with the Ring-LWE-based construction from~\cite{chen2019approximate}.
For a fair comparison, we re-parameterize the scheme in~\cite{chen2019approximate} such that the used secret has the same size with that in \Eagle{} and the overall size is nearly optimal for the target security level.
Nevertheless, for 80-bits (resp. 192-bits) of classical security level, the bandwidth of \Eagle{} is only about $30-40\%$ of that of the instantiation from~\cite{chen2019approximate}.
\Eagle{} is even smaller than Dilithium that is a representative LWE-based Fiat-Shamir signature scheme.
Detailed numbers are shown in Table~\ref{tab:lwe_cgm_comparison}.

\begin{table}[H]\renewcommand\arraystretch{1.2}
  \centering
	\caption{Comparisons in terms of sizes with Dilithium~\cite{dilithium} and~\cite{chen2019approximate}.
	The bit-security for Dilithium corresponds to
	the strongly-unforgeable version. The signature sizes for~\cite{chen2019approximate} and \Eagle{} are estimated as per the entropic bound.
	 }\label{tab:lwe_cgm_comparison}
	\setlength{\tabcolsep}{1.5mm}{}
  \begin{tabular}{|p{2.2cm}<{\centering} | p{2cm}<{\centering} | p{2cm}<{\centering} | p{2cm}<{\centering} |}
    \hline
                 	& Security (C/Q) & Pub. Key size (in bytes) & Sig. size  (in~bytes)  \\
    \hline
     Dilithium 1$^-$ &  89~/~81 & 992 & 1843 \\
     \cite{chen2019approximate}         & 79~/~71 & 2720 & 2753 \\
     \Eagle-512      & 79~/~71 & 928  & 1406 \\
     \hline
     Dilithium 3  &  176~/~159& 1952 & 3293 \\
     \cite{chen2019approximate}         & 180~/~164 & 7712 & 7172 \\
     \Eagle-1024      & 176~/~160 &1952& 3052 \\
    \hline
  \end{tabular}
\end{table}

%
%
%

\subsection{Comparison with \Robin}
As readers may have noticed, the Ring-LWE-based instantiation \Eagle{} is less efficient than the NTRU-based instantiation \Robin{} in Section~\ref{sec:ntru}.
For the NIST-III security level, while \Eagle{} and \Robin{} have roughly the same public key size, the \Eagle{} signatures are about 2 times the size of \Robin{} signatures.
This is an inherent gap, as the signatures in NTRU-based schemes are one ring element whereas the signatures in LWE-based schemes require at least two ring elements to recover the preimage.
In addition, the forgery security of \Eagle-1024 is lower than that of \Robin-1061 by more than 20-bits, although the degrees of the used ring are close.
The main cause is as follows.
While the public matrix $\matA$ in \Eagle{} is $n$-by-$3n$,
the best forgery attack would only use its submatrix of size $n$-by-$2n$, which is the same with the case of \Robin.
In constrast to \Robin, the acceptance bound of the \Eagle{} signature size is larger due to the wider $\matA$.
This lowers the forgery security.

Despite the worse performance than \Robin, \Eagle{} still occupies a fairly important position within the practical design of lattice signatures. The underlying Ring-LWE assumption could receive some preference to NTRU, especially for more powerful applications with overstretched parameters. Furthermore, \Eagle{} can be more conveniently adapted to the unstructured setting, thanks to the absence of costly matrix inversions in the key generation.
This may be a merit given the emphasis of post-quantum signatures not based on structured lattices raised by NIST\footnote{\url{https://csrc.nist.gov/csrc/media/Projects/pqc-dig-sig/documents/call-for-proposals-dig-sig-sept-2022.pdf}}.


\section{Conclusion}\label{sec:conclusion}
We develop a new lattice gadget construction of better compactness than the state-of-the-art. 
The main technique is a novel approximate gadget sampler, called semi-random sampler, in which the approximate error is deterministically generated and the preimage distribution is still simulatable without using the trapdoor. 
As an application, we present two practical hash-and-sign signature schemes instantiated with our compact gadget respectively based on NTRU and Ring-LWE. 
Our NTRU-based instantiation \Robin{} offers a quite simple implementation and high efficiency comparable to Falcon and Mitaka. 
This makes \Robin{} an attractive post-quantum signature for constrained environments. 
Our Ring-LWE-based scheme \Eagle{} is significantly smaller than 
the one~\cite{chen2019approximate} and even smaller than Dilithium. 
This demonstrates that LWE-based hash-and-sign signatures have much more potential than previously considered for practical applications. 

\subsection{Future works}
Our gadget framework actually supports diverse instantiations beyond the one used in \Robin{} and \Eagle. It would be interesting to explore more efficient constructions by combining different gadget matrices, lattice decoders and Gaussian samplers. 
It is also worthy to develop more algorithms for our gadget and then to build a complete toolkit as in~\cite{genise2019building}. 

Our gadget-based schemes are simpler than the NTRU trapdoor based ones and easily implemented fully over integers with the technique of~\cite{ducas2020integral}. 
We leave the optimized implementation and the provable side-channel protections as future works. 
In addition, our technique can be applied in advanced lattice cryptosystems. 
Evaluating its impact on the performance of advanced schemes needs a thorough investigation. 

Our proposals of \Robin{} and \Eagle{} do not fully integrate some recent techniques~\cite{jia2022lattice,espitau2022simpler,espitau2022shorter} to improve the performance and security, as we would like to focus more on the new gadget itself. 
Hence there shall be some room to improve the performance by adding these optimizations.


\bibliographystyle{alpha}
\bibliography{./biblio/ref}

\clearpage
\pagestyle{empty}
\appendix
~\vfill
\begin{center}
	\Huge\bfseries Supplementary Material
\end{center}
\vfill
\vfill
~
\clearpage
\section{Concrete Security Estimates}\label{appendix:securityestimate}
We estimate the concrete bit-security of our signature schemes according to the usual cryptanalytic methodology. 
To recap, we analyze the cost of the best attacks against key recovery and signature forgery, then translate the analysis into concrete bit-security using the Core-SVP model.  

\subsection{Lattice reduction and the Core-SVP model}
Lattice reduction is the task of finding a basis consisting of short and nearly orthogonal vectors. It is an important cryptanalytic tool used in lattice attacks. 
The most practical lattice reduction algorithms are BKZ~\cite{schnorr1994lattice} and its optimized variant~\cite{chen2011bkz,micciancio2016practical}. 
The BKZ algorithm is parameterized by the blocksize $\beta$. 
For a $d$-dimensional lattice $\lat$, BKZ-$\beta$ would find some $\vecv \in \lat$ with 
\[\|\vecv\| \leq \delta_\beta^d\vol(\lat)^{1/d}~~\text{and}~~
\delta_\beta \approx \left(\frac{(\pi\beta)^{\frac 1 \beta} \beta}{2\pi e}\right)^{\frac 1 {2(\beta-1)}}
\] 
for $\beta > 50$. 
The Core-SVP model estimates the cost of running BKZ-$\beta$ as $2^{0.292\beta}$ in the classical setting and $2^{0.265\beta}$ in the quantum setting. 
This is seen as a conservative concrete bit-security estimate. 

\subsection{Key recovery attack}
The key recovery against our signature schemes consists in finding the short secret $(\vecf,\vecg)$ such that $\matA\vecf + \vecg = \vecb \bmod Q$ where $\matA\in \Z_Q^{n\times n}$ and $\vecb \in \Z_Q^n$ are publicly known. 
The primal attack is a primary method for this task. 
It runs BKZ-$\beta$ on the lattice 
$\Lambda = \Lambda^{\perp}_Q([\matA \mid \matI_n \mid \vecb])$ 
of dimension $d = 2n+1$ to find the short $(\vecf, \vecg, -1)$. 
As shown in~\cite{alkim2016post}, a successful key recovery can be done when the blocksize $\beta$ satisfies 
\[\|(\vecf, \vecg, -1)\|\sqrt{\frac{3\beta}{4d}} \leq \delta_\beta^{2\beta - d - 1}\cdot Q^{\frac{n}{d}}\]
where $\sqrt{3/4}$ is set for a conservative estimate as in~\cite{falcon}. 
To optimize the attack, 
we also apply some known strategies prior to running BKZ:
\begin{enumerate}
	\item we guess the positions of $k$ zeros of $\vecf$ as suggested in~\cite{espitau2022shorter}; 
	\item we remove $l$ rows of $(\matA,\vecb)$ when constructing the lattice.  
\end{enumerate}
We choose $(k,l)$ to minimize the cost of the attack, which offers a few bits of improvement. 

\subsection{Forgery attack}
The signature forgery in our schemes is essentially to 
solve an approximate-CVP instance over the $q$-ary lattice. 
The nearest-colattice algorithm~\cite{espitau2020nearest} is a primary approximate-CVP algorithm. 
To address the twisted norm, we use the treatment of~\cite{espitau2022shorter} in the nearest-colattice framework. 
Specifically, given $(\matA,\vecu) \in \Z_Q^{n\times d} \times \Z_Q^n$ and $\beta > 0$, a preimage $(\vecx_0,\vecx_1) := \vecx \in \Z^n\times \Z^{d-n}$ such that $\matA \vecx = \vecu \bmod Q$ and $\|(\vecx_0,\gamma\vecx_1)\| \leq \beta$ can be computed by BKZ with blocksize $\beta$ satisfying 
\[
\beta \geq \min_{k\leq d-n}
	\big( \delta_\beta^{d-k}Q^\frac{n}{d-k}\gamma^{\frac{d-k-n}{d-k}} \big).
\]
We observe by experiments that for our LWE-based scheme, the best attack corresponds to $k \in [n, 2n]$, thus the associated CVP instance has a dimension $d-k \in [n,2n]$ as in the NTRU case. 
That is, the wider public matrix in the LWE setting does not seem to enhance the forgery security.

\end{document}